\documentclass[journal]{IEEEtran}

\usepackage{cite}
\usepackage{microtype} 
\usepackage{amsmath,amssymb,amsfonts}
\usepackage{amsthm}
\usepackage{tikz}
\usetikzlibrary{arrows.meta,backgrounds,bending,calc}
\usepackage{graphicx}
\usepackage{booktabs}
\usepackage{array}
\usepackage{algorithm}
\usepackage[noend]{algpseudocode}

\usepackage{xcolor}
\usepackage{flafter}   
\usepackage{enumitem}
\usepackage{xspace}
\usepackage[hidelinks]{hyperref}
\usepackage{url}
\usepackage{balance}   

\graphicspath{{diagrams/}}

\newcommand{\sys}{\textsc{Aegis}\xspace}
\newcommand{\acone}{\textsf{AC-1}\xspace}
\newcommand{\aconea}{\textsf{AC-1.a}\xspace}
\newcommand{\aconeb}{\textsf{AC-1.b}\xspace}
\newcommand{\actwo}{\textsf{AC-2}\xspace}
\newcommand{\heading}[1]{\par\smallskip\noindent\textbf{#1.}\enspace\ignorespaces}
\definecolor{crExposed}{RGB}{255,223,223}
\definecolor{crSealed}{RGB}{221,242,223}
\colorlet{crExposedLine}{red!70!black}
\colorlet{crSealedLine}{green!55!black}
\newcommand{\HW}{\mathsf{HW}}
\newcommand{\ExeInty}{\mathsf{ExeInty}}
\newcommand{\Iso}{\mathsf{Iso}}
\newcommand{\RAunf}{\mathsf{RA\text{-}unf}}
\newcommand{\Adv}{\mathsf{Adv}}
\newcommand{\negl}{\mathsf{negl}}
\newcommand{\leak}{\mathcal{L}}
\newcommand{\Faithful}{\mathsf{Faithful}}
\newcommand{\NonInt}{\mathsf{NonInt}}
\newcommand{\Verb}{\mathsf{Verbatim}}
\newcommand{\DestPol}{\mathsf{DestPolicy}}
\newcommand{\Pol}{\Delta}
\newcommand{\scheme}{\Pi_{\textsc{Aegis}}}

\theoremstyle{plain}
\newtheorem{theorem}{Theorem}
\newtheorem{lemma}{Lemma}
\newtheorem{definition}{Definition}

\begin{document}


\title{The Proxy Knows Too Much:\\ Sealing LLM API Routers with Attested TEEs}

\author{
Sipeng~Xie$^{1}$, Qianhong~Wu$^{1}$, Hengrun~Lu$^{1}$, Ziliang~Sun$^{1}$, Qi~Wu$^{1}$, Bo~Qin$^{2}$, Qin Wang$^{3}$
\\
\smallskip
$^{1}$\textit{Beihang University} $|$ $^{2}$\textit{Renmin University of China} $|$ $^{3}$\textit{Independent}



}


\maketitle

\begin{abstract}
Agents increasingly access large language models (LLMs) through API routers. A router terminates the client’s transport-layer security session and opens a separate upstream session, so it holds the full interaction in plaintext. This makes the router an application-layer man-in-the-middle: it can rewrite agent tool calls, swap dependencies for typosquatted packages, trigger attacks only under audit-evading conditions, and passively exfiltrate secrets. Existing client-side defenses are evadable.

We propose \sys, a provider-transparent attested API router whose data path is a client-verified faithful passthrough. \sys confines plaintext handling to a small hardware-enclave component while leaving authentication, scheduling, accounting, and management on the untrusted host. The client verifies the enclave before releasing plaintext. The host can neither read nor alter the interaction, and plaintext leaves only toward destinations fixed by the measured image. We show that all four malicious-router attack classes succeed against a plaintext-access baseline and are blocked by \sys, including adaptive tests against the same boundary. The trusted path is $851$ lines, carries three provider-native APIs without conversion, and completes every request under real-provider workload and concurrency. In a seeded audit pilot, two commodity coding agents find eight and ten of ten planted invariant violations. The local relay overhead is about six milliseconds per request.
\end{abstract}


\IEEEpeerreviewmaketitle

\section{Introduction}
\label{sec:intro}

Agents, including command-line coding assistants, increasingly reach model providers through an intermediary. They point at an LLM API router, a gateway that authenticates the caller, selects an upstream account, and forwards the request to one of several providers. A single endpoint multiplexes many providers and hides their keys behind one gateway key, with quota and billing centralized. The most widely deployed open-source router template draws tens of thousands of installations~\cite{yaim}.

This convenience hides a structural problem. A router terminates the client's transport-layer security (TLS) session and opens a separate one to the provider. The full request and response exist as plaintext inside it. It sees the user's prompts, the tool definitions and outputs the agent exchanges, the credentials in the traffic, and the tool-call payloads the provider returns. A router is an application-layer machine-in-the-middle by design, and the agent trusts it completely. Figure~\ref{fig:contrast} contrasts what the host observes on a plaintext router with what it observes through \sys.

\heading{The attack surface}
A recent measurement study formalizes what a malicious router can do from this position and documents real abuse in the wild~\cite{yaim}. First, the router can rewrite a tool call before it reaches the client. The rewrite is schema-valid and the client never sees the original, so one altered shell command is enough for remote code execution on the developer's machine. Second, it can swap a dependency for a typosquatted package inside an install command, slipping past allowlists that only check domains and planting a durable supply-chain foothold. Third, it can fire the rewrite only on a chosen trigger, such as a tool name, a keyword, or a request count; an auditor running a finite set of probes sees only benign behavior. Fourth, it can passively scan the traffic for credentials and exfiltrate them. The first three are response-side injection and its evasion variants; the fourth is passive request-path exfiltration. We label these attack classes \acone, \aconea, \aconeb, and \actwo, following the study.

The study also tests client-side defenses, including policy gates, anomaly screening, and a log of observed responses, and finds each evadable. It concludes that a sound fix must move response integrity to a party the client can trust, and that client-side mechanisms observing only the delivered responses cannot protect secrets on the request path. That conclusion creates a deployment gap. The trusted party is usually the provider, but clients and router operators cannot assume providers will add new APIs or infrastructure for them.

\begin{figure}[t]
\centering
\begingroup
\definecolor{crPanel}{RGB}{249,249,249}
\definecolor{crFrame}{RGB}{206,206,206}
\definecolor{crExposed}{RGB}{255,223,223}
\definecolor{crSealed}{RGB}{221,242,223}
\newcommand{\exposed}[1]{\setlength{\fboxsep}{1.2pt}\colorbox{crExposed}{\strut #1}}
\newcommand{\sealed}[1]{\setlength{\fboxsep}{1.2pt}\colorbox{crSealed}{\strut #1}}
\newcommand{\crbad}{\textcolor{red!70!black}{$\times$}}
\newcommand{\crok}{\textcolor{green!50!black}{$\checkmark$}}
\newcommand{\crpanel}[2]{%
  \setlength{\fboxsep}{4pt}\setlength{\fboxrule}{0.3pt}%
  \fcolorbox{crFrame}{crPanel}{%
    \begin{minipage}[t]{0.455\textwidth}
    {\footnotesize\bfseries #1}\par\vspace{3pt}
    {\scriptsize\ttfamily\selectfont\raggedright #2\par}
    \end{minipage}}}
\crpanel{\textcolor{red!50}{Plaintext router}}{%
  TLS terminates on the host\par\vspace{2pt}
  prompt: \exposed{fix the failing build} \crbad\par
  tool call: \exposed{run("npm install lodash")} \crbad\par
  api key: \exposed{sk-ant-a1b2c3d4} \crbad\par
  response: \exposed{assistant text + tool calls} \crbad\par
  \vspace{3pt}\hrule\vspace{3pt}
  the host reads and can rewrite \exposed{every field}}%
\hfill
\crpanel{\textcolor{teal}{Through \sys}}{%
  TLS terminates in the attested enclave\par\vspace{2pt}
  request: \sealed{0x9f2a3c...} (ciphertext) \crok\par
  tool call: \sealed{inside the ciphertext} \crok\par
  credential: \sealed{bound inside the enclave} \crok\par
  response: \sealed{0x7c1e8b...} (ciphertext) \crok\par
  \vspace{3pt}\hrule\vspace{3pt}
  the host sees \sealed{1240 tokens} of usage only}%
\endgroup
\caption{What the untrusted host observes for one coding-agent request. A plaintext router (top) holds every field in the clear; through \sys (bottom) the host sees only ciphertext bound to the attested enclave and usage telemetry.}
\label{fig:contrast}
\end{figure}

\heading{Sealing the router}
We accept that conclusion but move the trusted boundary to a place the client can verify without provider cooperation. Rather than detect a malicious router after the fact, can we remove the operator's ability to misbehave at all? We make the router an attested trust boundary. Only the data plane, the code path that carries the request and response bytes, runs inside a hardware enclave, and the client's TLS session terminates there. The host never holds the interaction in the clear. A small client-side component checks the enclave's measurement and refuses to send the body until the check passes. Authentication, account selection, and billing stay on the untrusted host and reach the enclave over a narrow channel that never carries the body. The provider sees an ordinary request, while the operator still schedules and bills but cannot read or modify the interaction, nor redirect it, since the provider and endpoint are fixed in the attested image.

\heading{One boundary suffices}
These four classes are not independent mechanisms to patch. They share one enabling capability, plaintext access at the router, and removing it collapses the family. A host that cannot see the interaction can neither rewrite a tool call nor its typosquat variant (\acone, \aconea), has no hidden behavior left for trigger-gating to conceal (\aconeb), and cannot scan the traffic for the client's secrets (\actwo), the last without any cooperation from the provider. The host still holds only the gateway credential it issued (\S\ref{sec:threat}).

\heading{Closing the gap}
Attested LLM gateways already exist~\cite{opengradient,redpill,phala}, but they attest a larger or transforming service rather than a faithful relay, and they verify only after the client has released its body. \sys closes both gaps by construction. It releases the body only after the client verifies the measurement and pins the enclave certificate, and it relays the interaction faithfully over a small reproducible trusted base, so the measurement certifies that the bytes were carried unmodified. \S\ref{sec:related} compares the two designs in detail. To our knowledge this is the first minimal trusted-code, reproducible enclave retrofit of a widely deployed open-source router template, evaluated against the attack classes above.

\heading{Contributions} We summarise our contributions as below.
\begin{itemize}[leftmargin=1.2em,itemsep=2pt,topsep=2pt]
\item A provider-transparent router design that closes the full malicious-router taxonomy (\acone, \aconea, \aconeb, \actwo) under a precise malicious-operator model, with stated coverage and scope (\S\ref{sec:threat}, \S\ref{sec:formal}).
\item A minimal trusted-code, confidentiality-first design. Only a faithful passthrough runs in the enclave, which keeps the trusted base small; the rest of the router stays on the untrusted host and reaches the enclave over a channel that never carries the body. The client refuses to release the interaction body until attestation verifies, so there is no after-the-fact check for the client to get wrong (\S\ref{sec:design}).
\item A measurement-equals-audited-code construction. A minimal reproducibly built trusted base and a signed build manifest, with the \textsf{ARPA} (Audit, Reproduce, Pin, Attest) bootstrapping protocol that lets an end user or their coding agent anchor trust without a central auditor (\S\ref{sec:arpa}).
\item An empirical evaluation on four fronts. All four attack classes succeed on a plaintext-access baseline and are blocked by our design; adaptive malicious-router tests and a scoped model in the ProVerif protocol verifier~\cite{proverif} bound the host's remaining influence; three provider-native workloads exercise the verified path; and seeded coding-agent audits check the auditability claim (\S\ref{sec:eval}).
\end{itemize}

\noindent We implement \sys on a widely available cloud enclave platform and report its security properties, provider support, auditability, and performance.

 \section{Technical Warm-ups}
\label{sec:bg}

\subsection{LLM API Routers and Tool-Using Agents}
An agent calls a language model in a loop, advertises tools the model may invoke, executes the tool calls it returns, and feeds the results back~\cite{react}. If a tool call is tampered with, the agent executes the tampered version, and how much harm follows depends on what its tools can reach. A coding agent is the sharpest case. Its tools run on the developer's machine and include shell commands and package installation~\cite{sweagent}, so a tool call it executes carries the same authority as the developer.

Many agents reach a provider through an intermediary that fronts one or more providers behind a single endpoint, an arrangement widely deployed as an LLM API gateway or aggregator~\cite{litellm,openrouter,portkey}. We call it an LLM API router. A client sends provider-format requests to the router under a credential the router issues. The router then authenticates the caller, selects an upstream provider account from a pool it manages, attaches that account's provider credential, forwards the request, relays the streamed response, and records token usage for accounting. Around this path deployments add load balancing and failover across the pool, per-account rate and concurrency limits, health-based selection, and session affinity. They expose every provider through one provider-compatible interface. An unmodified client switches providers by changing only the endpoint~\cite{litellm,portkey}. Routing here means brokering each request to an upstream account, not the model-selection routing that picks which model answers a query.

Brokering this way has one structural consequence the rest of the paper turns on. The router terminates the client's transport-layer security session and opens a separate one to the provider. By construction the full plaintext interaction, including the prompt, the response, and every tool input and output, exists inside the router for the duration of each request, the intended behavior of a terminating reverse proxy. From the outside, the agent cannot tell whether the router returned the provider's exact response or a substitute~\cite{yaim}.

\subsection{Trusted Execution and Remote Attestation}
A hardware trusted execution environment (TEE), or enclave, runs code in a processor-isolated region that even privileged host software, the operating system and the hypervisor included, cannot read or modify~\cite{sgxexplained}. Isolation alone does not tell a remote party which code is running. The platform therefore computes a measurement, a cryptographic hash of the exact image loaded into the enclave, and a hardware root of trust issues an attestation, a signed document that carries the measurement and traces to the platform vendor~\cite{awsnitroattest}. The remote party checks the signature and compares the measurement against the value it expects, learning that a specific image runs inside a genuine enclave. The platform also lets the enclave embed caller-supplied bytes in that document. Binding a fresh nonce and the enclave's transport key this way makes the attestation speak for the channel in use rather than some other session, the technique known as remote-attestation TLS (RA-TLS)~\cite{ratls}. \S\ref{sec:prelim} states this platform as an attested-execution scheme with three security games.

Two properties of this model matter for any system built on it. First, the trusted computing base (TCB), the code that must be correct for a guarantee to hold, spans both the code inside the enclave with the platform root and the relying party's verifier that checks the attestation and pins the measurement~\cite{flicker}. A guarantee is therefore only as strong as both parts, and a smaller enclave is worth the effort. Second, attestation proves code identity, not benignity, and a measurement means known-good code only when it ties back to audited, reproducibly built source~\cite{attestablebuilds}. Two platform limits also bound the model. An enclave that persists sealed state must defend against host rollback to an older valid version. Hardware isolation is not absolute either, since physical and microarchitectural side channels can leak enclave state~\cite{foreshadow,controlledchannel}.

\subsection{Cryptographic Primitives}
\label{sec:prelim}
We fix the cryptographic vocabulary the rest of the paper reuses, abstracting the enclave platform as an attested-execution scheme with three security properties and recalling the standard primitives the bootstrapping protocol relies on. Throughout, $\mathcal{A}$ is a probabilistic polynomial-time adversary, $\lambda$ the security parameter, and $\Adv^{X}\le\negl$ means the advantage is negligible in $\lambda$ for every such $\mathcal{A}$; for a guessing game the advantage is the distinguishing advantage $\big|\Pr[\beta'{=}\beta]-\tfrac12\big|$.

\heading{Attested execution}
We model the enclave platform as an attested-execution scheme, following the standard abstraction of trusted hardware~\cite{passshitramer}. A scheme is a tuple of probabilistic polynomial-time algorithms
\[
\HW=(\mathsf{Setup},\ \mathsf{Load},\ \mathsf{Run},\ \mathsf{Run\&Quote},\ \mathsf{VrfyQuote}).
\]
\begin{itemize}[leftmargin=1.2em,itemsep=2pt,topsep=2pt]
\item $\mathsf{Setup}(1^{\lambda})$ fixes the platform root key and outputs the \mbox{public parameters $\mathit{pms}$}.
\item $\mathsf{Load}(\mathit{pms},Q)$ creates an enclave running program $Q$, measures the loaded image to $m$, and returns a handle.
\item $\mathsf{Run}(\mathit{hdl},\mathit{in})$ runs $Q$ on an input.
\item $\mathsf{Run\&Quote}(\mathit{hdl},\mathit{in})$ additionally returns a quote $\rho$ binding the measurement, the input, the output, and caller-supplied bytes under the platform root.
\item $\mathsf{VrfyQuote}(\rho)$ checks a quote against the platform root.
\end{itemize}
Binding the enclave's transport key into the caller-supplied bytes makes a quote speak for the live session, the technique of attested transport-layer security. The scheme has three properties, each carrying one guarantee.

\begin{definition}[Execution integrity]
\label{def:exeinty}
$\Adv^{\ExeInty}_{\mathcal{A},\HW}\le\negl$, where $\mathcal{A}$ has oracle access to the
loaded enclave and wins by forcing an accepted output or streamed transcript
$\mathit{out}'\neq\mathsf{Run}(\mathit{hdl},\mathit{in})$ for some input. The loaded program
runs faithfully on its inputs, its whole input-output transcript included; this gives
execution integrity, not confidentiality.
\end{definition}

\begin{definition}[Memory confidentiality]
\label{def:iso}
$\Adv^{\Iso}_{\mathcal{A},\HW}\le\negl$ in the guessing game where $\mathcal{A}$ submits an
admissible pair of secret in-enclave inputs $s_0,s_1$ with $\leak(s_0)=\leak(s_1)$, the
challenger runs $Q$ on $s_\beta$, and $\mathcal{A}$ sees the entire host view. The host learns nothing of the enclave's secret memory beyond $\leak$, the isolation
property; this gives confidentiality.
\end{definition}

\begin{definition}[Attestation unforgeability]
\label{def:raunf}
$\Adv^{\RAunf}_{\mathcal{A},\HW}\le\negl$, where $\mathcal{A}$, given the public parameters and
a $\mathsf{Run\&Quote}$ oracle, wins by outputting an accepting quote on a
measurement-input-output tuple no enclave ever produced. No party without the platform root
can forge a quote.
\end{definition}

\heading{Standard primitives}
We recall four standard primitives:
\begin{itemize}[leftmargin=1.2em,itemsep=2pt,topsep=2pt]
\item A signature scheme is existentially unforgeable under chosen-message attack when no $\mathcal{A}$ with a signing oracle forges a valid signature on a fresh message \mbox{($\Adv^{\mathrm{euf\text{-}cma}}\le\negl$)}.
\item A measurement function $\mathsf{Build}$ maps a source or a build manifest under a fixed recipe $R$ to a measurement $m$ and is collision-resistant ($\Adv^{\mathrm{cr}}\le\negl$) when finding two distinct inputs with the same measurement is hard.
\item An authenticated channel, realized by authenticated encryption with associated data, delivers each message intact and confidential to its authenticated endpoint or signals failure, with distinguishing advantage $\Adv^{\mathrm{ch}}\le\negl$ from an ideal secure channel. The enclave-to-provider channel is endpoint-authenticated under the public certificate authorities, which we take as a trusted root; a certificate mis-issued for a pinned hostname is outside the model.
\item An append-only transparency log, under its public key $\mathit{pk}_{\mathrm{log}}$, binds its contents to a signed tree head $\eta$, proves a recorded entry with an inclusion proof $\pi$, and relates an earlier head $\eta_{-}$ to a later one with a consistency proof $\tau$. It has inclusion soundness and consistency ($\Adv^{\mathrm{log}}\le\negl$), so no $\mathcal{A}$ makes an honest verifier accept an inclusion proof for an entry the log does not record, nor makes two honest readers accept divergent contents at the same position.
\end{itemize}
 \section{System and Threat Model}
\label{sec:threat}

\subsection{System Model}
We target a self-hosted, single-operator router, where one entity owns the upstream provider accounts and the gateway logic. Three parties take part. A client runs a tool-using agent and a verifier that checks the router before trusting it. The operator runs the router on an untrusted host it fully controls. The providers are ordinary HTTPS upstreams the operator selects among. Several clients may share one operator, and the adversary is the operator.

The security question is whether clients can distinguish an operator running audited, attested code from one that merely asks to be trusted. Attestation gives clients that choice. A client can require the expected measurement before releasing plaintext, and organizations can make that check a routing requirement. An operator that publishes a reproducible measurement gives clients a verifiable trust signal, and operators without one stay indistinguishable from plaintext routers.

\begin{table*}[t]
\centering
\caption{What each trust boundary carries. The host-facing control channel (rows 2, 3, and 5) carries eighteen fields across two calls; no field can name a network address or carry body bytes.}
\label{tab:interface}
\footnotesize
\renewcommand{\arraystretch}{1.16}
\begin{tabular}{@{}lll@{}}
\toprule
\multicolumn{1}{@{}c}{\textbf{Boundary}} & \multicolumn{1}{c}{\textbf{Carries}} & \multicolumn{1}{c@{}}{\textbf{Never carries}} \\
\midrule
Sidecar $\to$ enclave & attested session; body plaintext after verification succeeds & plaintext before verification \\
Enclave $\to$ host (control) & gateway credential; model, platform, and session labels; failed accounts & body bytes; any destination \\
Host $\to$ enclave (control) & provider, policy, and account names; accounting label; provider credential & address or path; body transform \\
Enclave $\to$ provider & verbatim body; provider credential; allowlisted headers & gateway credential; host-chosen destination \\
Enclave $\to$ host (telemetry) & token counts, status, latency & prompt, response, tool content \\
\bottomrule
\end{tabular}
\end{table*}

\subsection{Threat Model}
Our adversary is a malicious router operator~\cite{yaim} that controls the untrusted host on which the router runs. It runs arbitrary code in the router process, reads and writes the router's datastores, observes and forwards all network traffic, and chooses which upstream account serves a request. Its goal is any of the four attack outcomes of \S\ref{sec:intro}, which Figure~\ref{fig:threat} places on the router's two plaintext paths.

\begin{figure}[b]
\centering
\newcommand{\actag}[1]{{\footnotesize\bfseries\sffamily\color{crExposedLine}#1}}%
\newcommand{\acsub}[1]{{\scriptsize\color{black!62}#1}}%
\resizebox{\columnwidth}{!}{%
\begin{tikzpicture}[
  font=\small,
  box/.style={rectangle, rounded corners=3pt, align=center, line width=1pt,
              draw=black!55, fill=black!4, inner sep=3pt, minimum height=1.5cm},
  mitm/.style={box, draw=crExposedLine, fill=crExposed, line width=1.4pt},
  data/.style={-{Stealth[length=5pt 1.4,width'=0pt 0.62]}, line width=1.0pt, color=black!85},
  lead/.style={-{Stealth[length=3.5pt 1.1,width'=0pt 0.6]}, line width=0.8pt, color=crExposedLine,
               dash pattern=on 0.5pt off 2pt},
  mlab/.style={font=\footnotesize, text=black!80, align=center, inner sep=1.5pt},
  callout/.style={rectangle, rounded corners=2.5pt, draw=crExposedLine, fill=crExposed,
                  align=left, inner sep=5pt, line width=0.9pt},
]
\def\yT{0.40}\def\yB{-0.40}
\node[box, minimum width=2.4cm] (ag) at (0,0)
  {\textbf{Agent}\\[1pt]{\scriptsize developer authority}};
\node[mitm, minimum width=3.1cm] (rt) at (4.55,0)
  {\textbf{Malicious router}\\[2pt]{\scriptsize terminates TLS;}\\[-1pt]{\scriptsize holds plaintext}};
\node[box, minimum width=2.4cm] (pv) at (9.1,0)
  {\textbf{Provider}\\[1pt]{\scriptsize honest}};
\draw[data] ([yshift=\yT cm]ag.east) -- ([yshift=\yT cm]rt.west);
\node[mlab,anchor=south] at (2.1,0.45) {prompts, tools,\\ secrets};
\draw[data] ([yshift=\yB cm]rt.west) -- ([yshift=\yB cm]ag.east);
\node[mlab,anchor=north] at (2.1,-0.45) {tool calls};
\draw[data] ([yshift=\yT cm]rt.east) -- ([yshift=\yT cm]pv.west);
\node[mlab,anchor=south] at (7.0,0.45) {fresh TLS};
\draw[data] ([yshift=\yB cm]pv.west) -- ([yshift=\yB cm]rt.east);
\node[mlab,anchor=north] at (7.0,-0.45) {response};
\node[callout,anchor=north] (ac) at (4.55,-1.55)
  {{\actag{AC-1}\, \acsub{rewrite a tool call}}\\[2pt]
   {\actag{AC-1.a}\, \acsub{swap in a typosquat package}}\\[2pt]
   {\actag{AC-1.b}\, \acsub{fire only on a trigger}}\\[2pt]
   {\actag{AC-2}\, \acsub{scan and exfiltrate secrets}}};
\draw[lead] (rt.south) -- (ac.north);
\end{tikzpicture}%
}
\caption{The attack position. All four attack classes start from the terminating router's plaintext access, response-side rewrites (\acone, \aconea, \aconeb) on one path and request scanning (\actwo) on the other.}
\label{fig:threat}
\end{figure}
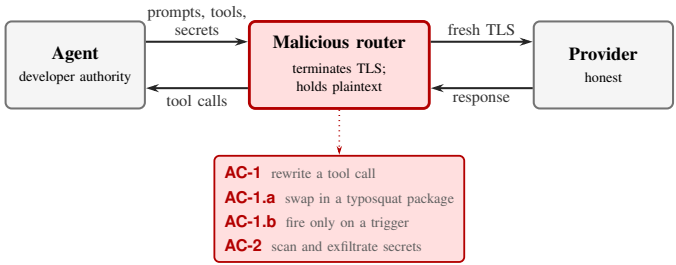

We assume the provider and the model are honest, the boundary the measurement study draws~\cite{yaim}. The provider need not attest, sign responses, expose a new API, or run new infrastructure, and remains an ordinary HTTPS upstream selected by the router. Defending against a malicious provider, or a model that emits a malicious tool call on its own, is out of scope for both works. We assume the enclave platform's hardware root of trust and attestation signing are sound, and place physical and microarchitectural side channels against the enclave out of scope.

\subsection{Security Goals}
\label{sec:goals}
\sys aims to let a malicious host carry the interaction without learning its body, altering it, or steering it off an approved path. Body-plaintext confidentiality hides the interaction body from the host and confines it to an approved destination. The gateway credential the caller authenticates with is disclosed to the host by design in the single-operator deployment. Router execution integrity fixes the code that touches the body to the attested image. Transport-level response integrity protects the bytes among the verified client, the enclave, and a validated provider. This covers the transported bytes rather than full semantic routing integrity, since account selection stays with the untrusted host.

\heading{Leakage profile}
Confidentiality holds modulo an explicit profile evaluated on the realized interaction transcript,
\begin{multline*}
\leak=(\,|b|,\ |r|,\ \mathrm{timing},\ \mathrm{cadence},\ \mathrm{status},\ \mathrm{abort/retry},\ g,\\
\mathrm{acct},\ \mathrm{prov},\ \mathrm{SNI},\ \mathrm{IP},\ \mathrm{rec\text{-}sizes},\ \mathrm{usage}\,),
\end{multline*}
splitting into an input-determined part fixed by the request, namely the body length $|b|$, the
request timing, the gateway credential $g$, the account and provider identifiers, the server-name
indication, the upstream address, and the request record sizes, and an interaction-determined part
fixed by the honest provider's reply, namely the response length $|r|$, the response timing, the
streaming cadence, the completion status, the abort-and-retry pattern, the response record sizes,
and the token usage. Treating the provider as a fixed oracle, $\leak$ is a function of the body, so
a pair is \emph{admissible} when its two bodies induce the same leakage. Without padding or traffic
shaping the design hides body content only up to $\leak$, not length or timing.

Four operational games specify body-confidential routing against a malicious host, each naming only the specific host strategy it rules out.

\begin{definition}[Confidentiality, equal-leakage]
\label{def:conf}
The host $\mathcal{A}$ chooses which image to load and carries a session on $b_\beta$, for a hidden bit $\beta$ and an admissible pair $b_0,b_1$; the client sidecar's attest-and-verify step releases plaintext only when the live enclave binds the pinned measurement and its transport key. For the host's guess $\beta'$,
\[
\Adv^{\mathrm{conf}}_{\mathcal{A},\scheme}
  = \bigl|\Pr\bigl\{\beta'{=}\beta\bigr\}-\tfrac12\bigr|
  \le \negl(\lambda).
\]
The host cannot tell which body was sent, and its freedom to load a deviating image or forge a binding is an event the experiment exposes.
\end{definition}

\begin{definition}[Destination authority]
\label{def:dest}
For a client body $b$ and application destinations $d$, provider hostname and path pairs,
\[
\Adv^{\mathrm{dest}}_{\mathcal{A},\scheme}
  = \Pr\bigl\{\,b \text{ in plaintext reaches } d\notin\Pol\,\bigr\}
  \le \negl(\lambda).
\]
No host strategy steers a plaintext body outside the baked set.
\end{definition}

\begin{definition}[Faithful relay]
\label{def:faithful}
For the submitted body $b$, the response $r$ the chosen provider returned, the body $b'$ a provider receives, and the response $r'$ the client accepts,
\[
\Adv^{\mathrm{int}}_{\mathcal{A},\scheme}
  = \Pr\bigl\{\,b'\neq b \,\vee\, r'\neq r\,\bigr\}
  \le \negl(\lambda).
\]
The host may force a request to fail, but cannot make a client accept a tampered success.
\end{definition}

\begin{definition}[Streaming integrity]
\label{def:stream}
\[
\Adv^{\mathrm{str}}_{\mathcal{A},\scheme}
  = \Pr\bigl\{\,\text{truncated response accepted}\,\bigr\}
  \le \negl(\lambda).
\]
A response completes only when its end-of-stream marker arrives intact, so a cut stream yields a rejection.
\end{definition}

These four games state that a malicious host learns the body only through $\leak$, cannot alter a delivered or accepted body, and cannot steer plaintext outside $\Pol$.
 \section{Our Solution: \sys}
\label{sec:design}

\sys turns the router into an attested trust boundary by splitting it along the line where plaintext lives, across the trust domains a request crosses (Figure~\ref{fig:arch}).

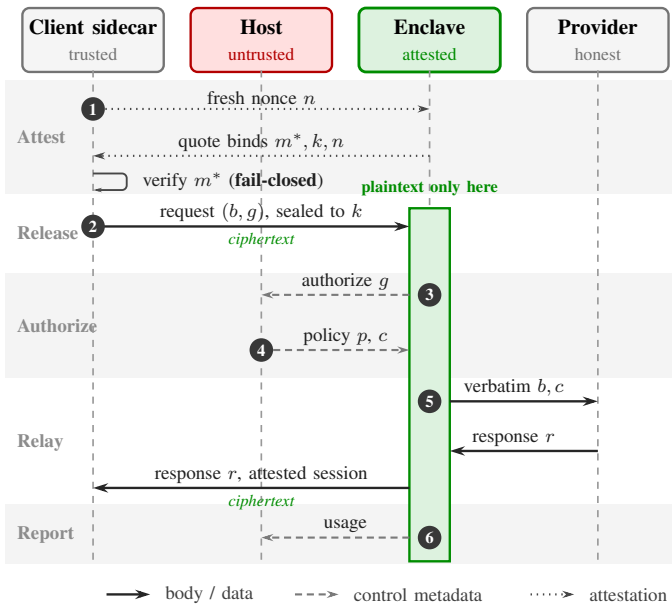
\begin{figure}[!]
\centering
\definecolor{bandA}{HTML}{F4F4F4}%
\resizebox{\columnwidth}{!}{%
\begin{tikzpicture}[
  font=\small, yscale=0.82,
  hdr/.style={rectangle, rounded corners=3pt, align=center, inner sep=2pt,
              minimum width=2.1cm, minimum height=0.95cm},
  lifeline/.style={dashed, line width=0.7pt, color=black!40},
  data/.style={-{Stealth[length=5pt 1.4,width'=0pt 0.62]}, line width=1.0pt, color=black!85},
  ctrl/.style={-{Stealth[length=4pt 1.2,width'=0pt 0.6]}, line width=0.9pt, color=black!52,
               dash pattern=on 3pt off 2pt},
  att/.style={-{Stealth[length=4pt 1.2,width'=0pt 0.6]}, line width=0.9pt, color=black!72,
              dash pattern=on 0.6pt off 2pt},
  selfcall/.style={-{Stealth[length=3pt 1.0,width'=0pt 0.5]}, line width=0.8pt,
                   color=black!72, rounded corners=2pt},
  mlab/.style={font=\footnotesize, text=black!88, inner sep=2pt},
  ctag/.style={font=\scriptsize\itshape, text=crSealedLine, inner sep=1.5pt},
  plab/.style={font=\footnotesize\bfseries, text=black!45},
  snum/.style={circle, fill=black!78, text=white, font=\scriptsize\bfseries,
               inner sep=0pt, minimum size=3.5mm},
]
\def\xS{0}\def\xH{2.5}\def\xE{5.0}\def\xP{7.5}
\def\FL{4.7}\def\FR{5.3}
\def\xL{-1.3}\def\xR{8.6}
\fill[bandA] (\xL,-0.72) rectangle (\xR,-2.78);
\fill[bandA] (\xL,-4.22) rectangle (\xR,-6.12);
\fill[bandA] (\xL,-8.42) rectangle (\xR,-9.50);
\node[plab,anchor=west] at (\xL+0.05,-1.75) {Attest};
\node[plab,anchor=west] at (\xL+0.05,-3.50) {Release};
\node[plab,anchor=west] at (\xL+0.05,-5.17) {Authorize};
\node[plab,anchor=west] at (\xL+0.05,-7.30) {Relay};
\node[plab,anchor=west] at (\xL+0.05,-8.96) {Report};
\node[hdr,draw=black!55,fill=black!4,line width=1pt] (hS) at (\xS,0)
  {\textbf{Client sidecar}\\[1pt]{\fontsize{7.5}{8.5}\selectfont\color{black!55}trusted}};
\node[hdr,draw=crExposedLine,fill=crExposed,line width=1pt] (hH) at (\xH,0)
  {\textbf{Host}\\[1pt]{\fontsize{7.5}{8.5}\selectfont\color{crExposedLine}untrusted}};
\node[hdr,draw=crSealedLine,fill=crSealed,line width=1.2pt] (hE) at (\xE,0)
  {\textbf{Enclave}\\[1pt]{\fontsize{7.5}{8.5}\selectfont\color{crSealedLine}attested}};
\node[hdr,draw=black!55,fill=black!4,line width=1pt] (hP) at (\xP,0)
  {\textbf{Provider}\\[1pt]{\fontsize{7.5}{8.5}\selectfont\color{black!55}honest}};
\draw[lifeline] (\xS,-0.55) -- (\xS,-9.45);
\draw[lifeline] (\xH,-0.55) -- (\xH,-9.45);
\draw[lifeline] (\xE,-0.55) -- (\xE,-9.45);
\draw[lifeline] (\xP,-0.55) -- (\xP,-9.45);
\fill[crSealed] (\xE-0.30,-9.45) rectangle (\xE+0.30,-3.05);
\draw[crSealedLine,line width=1pt] (\xE-0.30,-9.45) rectangle (\xE+0.30,-3.05);
\node[font=\scriptsize\bfseries,text=crSealedLine,anchor=south] at (\xE,-2.97)
  {plaintext only here};
\draw[att] (\xS,-1.25) -- (\xE,-1.25);
\node[mlab,above=0pt] at (\xH,-1.25) {fresh nonce $n$};
\node[snum] at (\xS,-1.25) {1};
\draw[att] (\xE,-2.10) -- (\xS,-2.10);
\node[mlab,above=0pt] at (\xH,-2.10) {quote binds $m^{\ast},k,n$};
\draw[selfcall] (\xS,-2.42) -- (0.5,-2.42) -- (0.5,-2.72) -- (\xS,-2.72);
\node[mlab,anchor=west] at (\xS+0.66,-2.56) {verify $m^{\ast}$ (\textbf{fail-closed})};
\draw[data] (\xS,-3.38) -- (\FL,-3.38);
\node[mlab,above=0pt] at (\xH,-3.38) {request $(b,g)$, sealed to $k$};
\node[snum] at (\xS,-3.38) {2};
\node[ctag] at (\xH,-3.62) {ciphertext};
\draw[ctrl] (\FL,-4.62) -- (\xH,-4.62);
\node[mlab,above=0pt] at (3.75,-4.62) {authorize $g$};
\node[snum] at (\xE,-4.62) {3};
\draw[ctrl] (\xH,-5.62) -- (\FL,-5.62);
\node[mlab,above=0pt] at (3.75,-5.62) {policy $p$, $c$};
\node[snum] at (\xH,-5.62) {4};
\draw[data] (\FR,-6.55) -- (\xP,-6.55);
\node[mlab,above=0pt] at (6.25,-6.55) {verbatim $b,c$};
\node[snum] at (\xE,-6.55) {5};
\draw[data] (\xP,-7.45) -- (\FR,-7.45);
\node[mlab,above=0pt] at (6.25,-7.45) {response $r$};
\draw[data] (\FL,-8.12) -- (\xS,-8.12);
\node[mlab,above=0pt] at (\xH,-8.12) {response $r$, attested session};
\node[ctag] at (\xH,-8.38) {ciphertext};
\draw[ctrl] (\FL,-9.02) -- (\xH,-9.02);
\node[mlab,above=0pt] at (3.75,-9.02) {usage};
\node[snum] at (\xE,-9.02) {6};
\def\yLg{-10.05}
\draw[data] (0.2,\yLg) -- (0.85,\yLg);
\node[anchor=west,font=\footnotesize,text=black!75] at (0.95,\yLg) {body / data};
\draw[ctrl] (3.0,\yLg) -- (3.65,\yLg);
\node[anchor=west,font=\footnotesize,text=black!75] at (3.75,\yLg) {control metadata};
\draw[att] (6.5,\yLg) -- (7.15,\yLg);
\node[anchor=west,font=\footnotesize,text=black!75] at (7.25,\yLg) {attestation};
\end{tikzpicture}%
}
\caption{The \sys architecture across three trust domains. Numbered steps trace a request's lifecycle (Algorithms~\ref{alg:bootstrap}--\ref{alg:relay}).}
\label{fig:arch}
\end{figure}

\subsection{Overview}
What belongs inside the enclave fixes everything else, and the two extremes fail oppositely. Placing the entire transforming router inside protects the body but swells the trusted base to the whole application; its measurement then names a large, fast-moving image that no end user can tie to audited source (\S\ref{sec:arpa}). Keeping the body path on the host and attesting only a side component leaves the host terminating the client's session, so the plaintext access behind every attack survives. \sys takes the point between them, a construction of probabilistic polynomial-time algorithms
\[
\scheme=(\mathsf{Publish},\ \mathsf{Audit},\ \mathsf{Pin},\ \mathsf{Attest},\ \mathsf{Relay},\ \mathsf{Vrfy}).
\]
\begin{itemize}[leftmargin=1.2em,itemsep=2pt,topsep=2pt]
\item $\mathsf{Publish}$, run by the operator once per release, publishes an audited, reproducible build.
\item $\mathsf{Audit},\mathsf{Pin}$, run once per client, audit the published build and pin its measurement.
\item $\mathsf{Attest},\mathsf{Vrfy}$, run per session, attest the live enclave to the client sidecar before any plaintext leaves the client.
\item $\mathsf{Relay}$, run per request, relays the interaction verbatim to a fixed provider and reports only usage to the host.
\end{itemize}

\subsection{Minimal Trusted Base}
The trusted base should contain only what confidentiality and integrity require, because a smaller base is both a smaller attack surface and, per \S\ref{sec:arpa}, a feasible audit target. We therefore keep three things in the enclave and nothing else, namely termination of the client session, faithful relay of the request and response, and the set of legitimate upstream destinations. Everything else stays on the host, including authentication, account selection and scheduling, credential storage, usage accounting, and the management interface. The host reaches the enclave over a narrow control channel, fixed in Table~\ref{tab:interface}, with two calls. One authenticates the caller and returns the routing decision, in which the destination appears only as a provider and a policy name. The other reports usage after the request completes and carries only numeric counters that cannot encode text. No field on either call can name a network address or carry a byte of the body. The enclave keeps no state across requests, so it writes nothing to disk and rollback attacks against stateful enclaves do not arise.

\subsection{Anchoring the Measurement to Audited Source}
\label{sec:arpa}
Bootstrapping comes first in time. Attestation proves which code an enclave runs, so a measurement is trustworthy only when that code ties back to audited source. The operator commits a release with $\mathsf{Publish}$, and each client runs $\mathsf{Audit},\mathsf{Pin}$ (Algorithm~\ref{alg:bootstrap}) to audit and reproducibly rebuild that source, pinning the resulting measurement $m^{\ast}$ only when the rebuild matches the manifest (lines~\ref{ln:rebuild}--\ref{ln:pin}) and the manifest's entry verifies against a public transparency log (lines~\ref{ln:incl}--\ref{ln:cons}). The pinned measurement is therefore one the operator has committed to in the open, and a forked or rolled-back log fails closed. Pinning audited source closes the malicious-operator case by construction. The audit rests on review rather than proof. It raises the bar against a malicious author without ruling that case out. The code that first receives the plaintext is the code the client audited, and a measurement one client trusts is one every client can see in the public log. Because the trusted code is a small faithful passthrough, one end user or their coding agent can run this audit without a central auditor.

\begin{algorithm}[!ht]
\caption{Bootstrapping. The operator runs $\mathsf{Publish}$ once per release; each client then runs $\mathsf{Audit},\mathsf{Pin}$ once.}
\label{alg:bootstrap}
\begin{algorithmic}[1]
\Require $s$ source, $R$ recipe, $\mathit{mf}$ manifest, $\sigma$ signature, $\pi$ inclusion proof, $\eta$ tree head, $\tau$ consistency proof, $\eta_{-}$ last-seen head
\Procedure{$\mathsf{Publish}$}{$s,R$}\Comment{operator}
  \State $m \gets \mathsf{Build}(s,R)$
  \State $\mathit{mf} \gets (s,R,m)$;\ \ $\sigma \gets \mathsf{Sign}_{\mathit{sk}_{\mathcal{O}}}(\mathit{mf})$
  \State $(\pi,\eta) \gets \mathsf{LogAppend}(\mathit{mf})$
  \State \Return $(\mathit{mf},\sigma,\pi,\eta)$
\EndProcedure
\Procedure{$\mathsf{Audit},\mathsf{Pin}$}{$\mathit{mf},\sigma,\pi,\eta,\tau$}\Comment{client}
  \State $(s,R,m) \gets \mathit{mf}$
  \State $\mathbf{assert}\ \mathsf{SigVrfy}_{\mathit{pk}_{\mathcal{O}}}(\sigma;\mathit{mf}) \wedge \Faithful(s)$
  \State $\mathbf{assert}\ \mathsf{VrfyIncl}_{\mathit{pk}_{\mathrm{log}}}(\mathit{mf};\pi,\eta)$ \label{ln:incl}
  \State $\mathbf{assert}\ \mathsf{VrfyCons}_{\mathit{pk}_{\mathrm{log}}}(\eta_{-},\eta;\tau)$ \label{ln:cons}
  \State $\mathbf{assert}\ \mathsf{Build}(s,R){=}m$ \label{ln:rebuild}
  \State $m^{\ast} \gets m$;\ \ $\eta_{-} \gets \eta$ \label{ln:pin}
\EndProcedure
\end{algorithmic}
\end{algorithm}

\subsection{Fail Closed Before Releasing Plaintext}
\label{sec:failclosed}
With a measurement pinned, each session begins by verifying the live enclave. The client points its agent at a thin local proxy, the sidecar, whose $\mathsf{Vrfy}$ procedure (Algorithm~\ref{alg:attest}, lines~\ref{ln:vq}--\ref{ln:release}; step~1 of Figure~\ref{fig:arch}) enforces confidentiality. It releases the request body and gateway credential (line~\ref{ln:release}) only after the quote passes the platform-root and debugging checks (line~\ref{ln:vq}) and binds the pinned measurement $m^{\ast}$, the fresh nonce, and the live session's certificate (line~\ref{ln:vbind}), and it forwards nothing on any failure. Because plaintext reaches only a verified enclave, the host never sees the body.

This ordering is the load-bearing choice. It separates two postures, refuse-before-send and verify-after-receive. A verify-after-receive router has already been given the plaintext. A client that follows a multi-step manual check can get a step wrong and believe it verified more than it did. \sys gives the client exactly one success path and makes releasing plaintext depend on it. There is no partial receipt to misread and no after-the-fact step to skip, and no signature is carried inside the response it is meant to protect. \S\ref{sec:related} returns to why this removes a class of verification mistakes that a verify-after-receive design invites.

\begin{algorithm}[t]
\caption{$\mathsf{Attest},\mathsf{Vrfy}$, the per-session attestation handshake (step~1 of Figure~\ref{fig:arch}).}
\label{alg:attest}
\begin{algorithmic}[1]
\Require $n$ nonce, $\rho$ quote, $b$ request body, $g$ gateway credential, $m^{\ast}$ pinned measurement, $\mathit{cert}_{\mathcal{E}}$ channel certificate
\Procedure{$\mathsf{Attest}$}{$n$}\Comment{enclave}
  \State $\rho \gets \mathsf{Run\&Quote}(\mathit{hdl},(n,H(\mathit{cert}_{\mathcal{E}})))$
  \State \Return $\rho$
\EndProcedure
\Procedure{$\mathsf{Vrfy}$}{$\rho,n,b,g$}\Comment{client sidecar}
  \State $\mathbf{assert}\ \mathsf{VrfyQuote}_{\mathit{pk}_{\mathrm{root}}}(\rho) \wedge \neg\mathsf{dbg}(\rho)$ \label{ln:vq}
  \State $\mathbf{assert}\ m_{\rho}{=}m^{\ast} \wedge n_{\rho}{=}n \wedge h_{\rho}{=}H(\mathit{cert}_{\mathcal{E}})$ \label{ln:vbind}
  \State $\mathsf{send}_{\mathit{cert}_{\mathcal{E}}}(b,g)$ \label{ln:release}
\EndProcedure
\end{algorithmic}
\end{algorithm}

Binding integrity to the live session rather than to a receipt also protects streamed responses byte by byte. A coding agent consumes a streamed response incrementally and may act on an early tool-call chunk before the response completes. An integrity check computed at end of stream arrives too late to gate that action. Because \sys carries every byte inside the attested session, each chunk is integrity-protected as it arrives. An agent that acts on an early chunk acts on bytes the host could neither read nor alter.

\subsection{The In-Enclave Relay}
\label{sec:relay}
Once a session is verified, the released body reaches $\mathsf{Relay}$ (Algorithm~\ref{alg:relay}, steps~3--6 of Figure~\ref{fig:arch}), the only code that ever holds the plaintext (line~\ref{ln:dec}). The relay performs no semantic transformation. It forwards the request body to the provider unchanged under the host-selected credential, dropping the caller's gateway credential and forwarding only an allowlisted set of headers (line~\ref{ln:fwd}), relays the response bytes unchanged (line~\ref{ln:stream}), and extracts token usage for the host to bill. It does not rewrite tool calls, reshape multimodal content, convert between protocols, or log bodies. From the provider's perspective the enclave is the router endpoint making the request the router would have made. We call this a faithful passthrough. The end-of-stream marker is forwarded only on upstream completion (line~\ref{ln:eos}), so a truncated upstream response reaches the client visibly incomplete.

Faithfulness is a design commitment that serves three ends. It preserves agent semantics, because a coding agent depends on the provider's exact tool calls and streamed structure, and any rewrite risks breaking the loop. It keeps the trusted base small, because a verbatim relay is far less code than a transformer. And it closes a gap that transformation opens. If the attested code reshapes content, attestation and any response signature prove only that the mediating code ran, not that the client's exact request reached the provider and the provider's exact response reached the client. A faithful passthrough makes the attested code running equivalent to the bytes being carried unmodified, the property the client needs. The equivalence rests on the audit of \S\ref{sec:arpa} and the conformance tests of \S\ref{sec:eval}. A verbatim relay also keeps the host from altering any byte between the verified enclave and the provider, and the client receives the provider's exact response.

\begin{algorithm}[t]
\caption{$\mathsf{Relay}$, the enclave data path, run per request.}
\label{alg:relay}
\begin{algorithmic}[1]
\Require $\mathit{hdl}$, $\mathit{ct}$, $k$, $\Pol$, accounting label $\ell$, failover bound $\beta$
\Ensure ciphertext response stream, usage report
\State $(b,g) \gets \mathsf{Dec}_{k}(\mathit{ct})$ \label{ln:dec}
\State $F \gets \emptyset$
\Repeat
  \State $(\mathit{prov},p,c) \gets \mathsf{Authorize}(g,\ell,F)$
  \State $d \gets \Pol[p]$ \label{ln:resolve}
  \State $\mathbf{assert}\ d{\neq}\bot$ \label{ln:reject}
  \State $\mathit{ok} \gets \mathsf{send}_{d}(b,\,\mathrm{hdr}\oplus c\ominus g)$ \label{ln:fwd}
  \If{$\neg\mathit{ok}$}
    \State $F \gets F \cup \{\mathit{prov}\}$
  \EndIf
\Until{$\mathit{ok} \;\vee\; |F|{=}\beta$}
\ForAll{$r_i \in \mathsf{strm}(r)$}
  \State $\mathbf{emit}\ \mathsf{Enc}_{k}(r_i)$ \label{ln:stream}
\EndFor
\If{$\mathsf{eos}(r)$}
  \State $\mathbf{emit}\ \mathsf{eos}(r)$ \label{ln:eos}
\EndIf
\State $\mathbf{report}\ \mathsf{use}(r)$
\end{algorithmic}
\end{algorithm}

\subsection{Enclave-Owned Destination}
\label{sec:dest}
The host names the route but cannot choose the address. The enclave owns the set of legitimate destinations $\Pol$, baked into the measured image, and the host's routing decision references a provider and an endpoint policy $p$ by name only (step~4 of Figure~\ref{fig:arch}). The relay resolves that name to a fixed destination $d=\Pol[p]$ the host cannot influence (Algorithm~\ref{alg:relay}, line~\ref{ln:resolve}), and it rejects a decision naming anything outside $\Pol$ (line~\ref{ln:reject}). Plaintext therefore leaves only over a session the enclave authenticates for a policy hostname, and the host cannot redirect it to an address of its choosing. For a first-party upstream the enclave pins the provider hostname and validates its certificate against the public certificate authorities, since it cannot attest a provider that is not an enclave. These trust roots live in the measured image the client verifies. When the next hop is another confidential router, the enclave can instead require attested transport and an allowlisted next-hop measurement, the attested-upstream direction \S\ref{sec:discussion} develops as future work. \section{Formal Analysis}
\label{sec:formal}

The guarantees stated plainly in \S\ref{sec:design}, body confidentiality (\S\ref{sec:failclosed}), faithful relay and response integrity (\S\ref{sec:relay}), and destination authority (\S\ref{sec:dest}), are ordering and binding claims about what must already have happened before plaintext is released, dispatched toward a provider, or accepted back. We make each precise with a game-based reduction that bounds the data-path adversary's advantage under standard cryptographic assumptions, and two paper-and-pencil lemmas anchor the runtime measurement to the source the client audited (\S\ref{sec:arpa}) and characterize what an honest verifier accepts. A machine-checked symbolic model independently cross-checks the protocol ordering and falsifies each check removal; we give it as supporting analysis in the supplementary material, since its secrecy property overlaps the confidentiality theorem rather than reproving it. The full proofs are there as well.

\subsection{Reduction Theorems}
\label{sec:formal-thms}

The analysis is over the construction $\scheme$ of \S\ref{sec:design}, with $\mathsf{Relay}$ the in-enclave data path of Algorithm~\ref{alg:relay} and $Q$ its loaded program. Each theorem names only the assumptions its property consumes, drawn from the attested-execution scheme $\HW$ and primitives of \S\ref{sec:prelim}, the faithfulness predicate below, and the leakage profile $\leak$, and bounds the corresponding advantage by a sum of the advantages an adversary would have to win to break it. Existential unforgeability and collision resistance do not appear here; they enter only the anchoring lemma below, which justifies treating $Q$ as the audited one.

\heading{The faithfulness predicate}\label{sec:faithful}
The relay protects the body only because the loaded program forwards it without leaking it; we state this as a predicate and decompose it into three machine-checkable parts. A loaded program $Q$ is faithful when
\begin{multline*}
\Faithful(Q)=\NonInt(Q,\leak)\\
\wedge\ \Verb(Q)\ \wedge\ \DestPol(Q,\Pol).
\end{multline*}
\begin{itemize}[leftmargin=1.2em,itemsep=2pt,topsep=2pt]
\item Host non-interference $\NonInt(Q,\leak)$ requires that every host-observable output of $Q$ be a function of $\leak$ alone, leaking no body-derived information through any host-facing channel, including logs and error or abort messages, outbound transport headers and request-line fields, metrics and control-message fields, name resolution and server-name indication, and transport framing, flush schedule, connection reuse, and retry behavior.
\item Verbatim relay $\Verb(Q)$ requires that $Q$ forward the request body and return the response body byte for byte, including every streamed chunk through the end-of-stream marker, with no transformation path.
\item Destination policy $\DestPol(Q,\Pol)$ requires that $Q$ resolve every destination inside the baked set $\Pol$, taking from the host a policy name and never a network address or path.
\end{itemize}
Each enumerated channel maps to a concrete evidence item \S\ref{sec:eval} supplies, and full machine-checked verification of $\NonInt$ is future work. We write $\Faithful(s)$ for $\Faithful(Q)$ on the program $Q$ that source $s$ builds to under the fixed recipe.

\begin{theorem}[Confidentiality]
\label{thm:conf}
Consider the equal-leakage confidentiality experiment in which the host chooses which image
to load and the client sidecar runs the attest-and-verify step before any plaintext is
released. Assume $\HW$ has memory confidentiality with respect to $\leak$ and execution
integrity, has remote-attestation unforgeability, the client-to-enclave and enclave-to-provider
channels are secure channels, and the loaded data path $Q$ satisfies $\NonInt(Q,\leak)$. Then
\[
\begin{aligned}
\Adv^{\mathrm{conf}}_{\mathcal{A},\scheme}(\lambda)
&\le \Adv^{\Iso}_{\mathcal{A},\HW}(\lambda)+\Adv^{\ExeInty}_{\mathcal{A},\HW}(\lambda)\\
&\phantom{{}\le{}}+\Adv^{\RAunf}_{\mathcal{A},\HW}(\lambda)+\Adv^{\mathrm{ch}_{\mathrm{client}}}_{\mathcal{A}}(\lambda)\\
&\phantom{{}\le{}}+\Adv^{\mathrm{ch}_{\mathrm{prov}}}_{\mathcal{A}}(\lambda).
\end{aligned}
\]
\end{theorem}

\begin{theorem}[Faithful relay and integrity]
\label{thm:faithful}
Assume $\HW$ has execution integrity, the loaded data path $Q$ satisfies $\Verb(Q)$ and
emits the end-of-stream marker only on upstream completion, and the client-to-enclave and
enclave-to-provider channels are secure channels. Then
\[
\begin{aligned}
\Adv^{\mathrm{int}}_{\mathcal{A},\scheme}(\lambda)
&\le \Adv^{\ExeInty}_{\mathcal{A},\HW}(\lambda)\\
&\phantom{{}\le{}}+\Adv^{\mathrm{ch}_{\mathrm{client}}}_{\mathcal{A}}(\lambda)+\Adv^{\mathrm{ch}_{\mathrm{prov}}}_{\mathcal{A}}(\lambda).
\end{aligned}
\]
The same bound applies to $\Adv^{\mathrm{str}}_{\mathcal{A},\scheme}$; the provider-channel term
is needed because a host that forges an end-of-stream marker on the provider channel would have
the faithful $Q$ relay it verbatim. The host may force a request to fail and may cut a stream,
but cannot make the client accept a tampered chunk or treat a truncated stream as a completed
response.
\end{theorem}

\begin{theorem}[Destination authority]
\label{thm:dest}
Assume $\HW$ has execution integrity, the loaded data path $Q$ satisfies $\DestPol(Q,\Pol)$,
and the enclave-to-provider channel is a secure channel endpoint-authenticated under the
public certificate authorities, a trusted root. Then
\[
\Adv^{\mathrm{dest}}_{\mathcal{A},\scheme}(\lambda)\ \le\
\Adv^{\ExeInty}_{\mathcal{A},\HW}(\lambda)
+\Adv^{\mathrm{ch}_{\mathrm{prov}}}_{\mathcal{A}}(\lambda).
\]
The guarantee is application-layer, not network-layer egress confinement. The host carries
every packet and may choose the transport-layer peer, but plaintext leaves only over a
session the enclave authenticates for a policy hostname.
\end{theorem}

The three theorems share the assumption that the loaded program $Q$ is the audited one and is faithful. The anchoring lemma below discharges the first half, tying the runtime measurement to source the client rebuilt; faithfulness rests on the predicate decomposition of \S\ref{sec:design}, supported by the implementation evidence of \S\ref{sec:eval} and named as future machine-checked work for full information-flow verification.

\subsection{Anchoring and Accountability}
\label{sec:formal-lemmas}

\begin{lemma}[Measurement anchoring]
\label{lem:anchor}
Assume the build is deterministic, the measurement function $\mathsf{Build}$ is
collision-resistant over loaded images, the operator's signature scheme is existentially
unforgeable under chosen-message attack, and the platform root signs evidence only for
images loaded into an enclave. If a client completes the audit-and-pin phase of
\S\ref{sec:arpa} over source $s$ and pins $m^{\ast}=\mathsf{Build}(s,R)$, then any
plaintext its sidecar releases reaches only an enclave running the image built from the
audited source $s$. Concretely,
\[
\Adv^{\mathrm{anch}}_{\mathcal{A}}(\lambda)\ \le\
\Adv^{\RAunf}_{\mathcal{A},\HW}(\lambda)+\Adv^{\mathrm{cr}}_{\mathcal{A}}(\lambda)
+\Adv^{\mathrm{euf\text{-}cma}}_{\mathcal{A}}(\lambda),
\]
where the experiment returns $1$ when a released plaintext reaches an enclave not running that image.
\end{lemma}

Anchoring composes with the three theorems. The lemma fixes which code first receives the plaintext, and the theorems bound what the host learns of and does to that plaintext while the audited code runs. The accountability lemma below characterizes what an honest verifier accepts when it binds a live measurement, and it is the second place existential unforgeability and collision resistance enter. The verifier binds the full build manifest of \S\ref{sec:arpa}, not a bare source hash.

\begin{lemma}[Accountability]
\label{lem:account}
Assume the operator's signature scheme is existentially unforgeable under chosen-message
attack, the measurement function $\mathsf{Build}$ is collision-resistant, and the
append-only transparency log is inclusion-sound and consistent, so an accepted inclusion
proof implies a recorded entry and no two honest verifiers read divergent contents at the
same position. Then an honest verifier accepts a live measurement only
when that measurement is bound to source that the operator has signed, that is recorded in
the transparency log, and that reproducibly rebuilds to the measurement. Concretely,
\[
\Adv^{\mathrm{acc}}_{\mathcal{A}}(\lambda)\ \le\
\Adv^{\mathrm{euf\text{-}cma}}_{\mathcal{A}}(\lambda)
+\Adv^{\mathrm{cr}}_{\mathcal{A}}(\lambda)
+\Adv^{\mathrm{log}}_{\mathcal{A}}(\lambda),
\]
where the experiment returns $1$ iff an honest verifier accepts a measurement that is not
so bound.
\end{lemma}

\heading{Scope}
The reductions cover the data path under standard cryptographic assumptions and treat byte-for-byte faithfulness and the control-channel schemas as the named assumptions of \S\S\ref{sec:prelim} and~\ref{sec:faithful}; neither covers the honesty of the audited source, the residual the review of \S\ref{sec:arpa} narrows but cannot close. \S\ref{sec:eval} grounds the faithfulness predicates in implementation evidence, with a conformance suite for verbatim relay, source-derived schema tests for the body-free control channel, and a live ciphertext capture for the host hop.
 \section{Implementation and Evaluation}
\label{sec:eval}

\subsection{Prototype}
We implement \sys on AWS Nitro Enclaves~\cite{awsnitroattest,awsnitrotls}, a widely available cloud enclave that isolates a virtual machine from its host and reaches it only over a virtual socket, with a hardware-signed attestation over the loaded image. Three binaries realize the split of \S\ref{sec:design}. The data plane runs inside the enclave under the nitriding toolkit~\cite{nitriding}, the host runs the rest of the router, and the client runs the verifying sidecar. We enforce the boundary at build time and fail the build if the data plane links any host service-layer package.

\subsection{Evaluation Setup}
The evaluation separates security properties from runtime cost. The security tests run against the boundary the deployed image exposes. The runtime measurements use an EC2 c5.4xlarge instance (Intel Xeon Platinum 8275CL at $3.00$~GHz, $16$ vCPUs over $8$ physical cores), with the enclave allocated $2$ dedicated vCPUs on $1$ physical core and $3$~GiB of memory. The runs use two enclave images. Every security test and every real-provider measurement uses the deployed security image through a verified sidecar and enclave. The steady-state overhead runs instead use a microbenchmark image that repoints the destination policy at a local upstream with pooled connections and trusts its certificate, which isolates the relay's own cost from provider variance. This image has a different measurement and is not the security artifact. Per-run sample sizes appear in the figure captions.

\subsection{Security-Property Evaluation}
We evaluate \sys through deterministic security properties, the capabilities the design removes, rather than as a statistical detector. Table~\ref{tab:harness} is the central result. First, we replay each attack class against a plaintext-access baseline, the ordinary posture of a terminating router that the measurement study documents. Second, we mutate the host-controlled interfaces of \sys and check that each mutation is either outside the enclave's authority or fails closed. The baseline column confirms each adversary action succeeds on a plaintext router; the \sys rows therefore measure a removed capability rather than an untried one, and the \sys column attributes each block to a concrete guard.

\begin{table*}[t]
\centering
\caption{Attack and guard coverage. Attack classes reproduce on a plaintext-access baseline; malicious-host mutations against \sys are attributed to fail-closed guards.}
\label{tab:harness}
\footnotesize
\renewcommand{\arraystretch}{1.16}
\newcommand{\guardok}[1]{\textcolor{green!45!black}{$\checkmark$~#1}}
\newcommand{\attackbad}[1]{\textcolor{red!70!black}{$\times$~#1}}
\newcommand{\attackwarn}[1]{\textcolor{orange!85!black}{$\blacktriangle$~#1}}
\begin{tabular}{@{}>{\raggedright\arraybackslash}p{0.24\textwidth} >{\raggedright\arraybackslash}p{0.16\textwidth} >{\raggedright\arraybackslash}p{0.16\textwidth} >{\raggedright\arraybackslash}p{0.34\textwidth}@{}}
\toprule
\multicolumn{1}{@{}c}{\textbf{Adversary action}} & \multicolumn{1}{c}{\textbf{Plain router}} & \multicolumn{1}{c}{\textbf{\sys}} & \multicolumn{1}{c@{}}{\textbf{Guard exercised}} \\
\midrule
\acone\ tool-call rewrite & \attackbad{attack succeeds} & \guardok{blocked} & Host has no response-body rewrite point. \\
\aconea\ typosquat rewrite & \attackbad{attack succeeds} & \guardok{blocked} & Faithful relay; package names are never rewritten. \\
\aconeb\ trigger-gated rewrite & \attackbad{attack succeeds} & \guardok{blocked} & Rewrite impossible at source; trigger is moot. \\
\actwo\ secret scan & \attackbad{attack succeeds} & \guardok{blocked} & Body-free host RPC and ciphertext host hop. \\
Unknown policy or provider & \attackbad{host-chosen route} & \guardok{fails closed} & Enclave resolves only baked policy names. \\
Redirect to attacker host & \attackbad{attack succeeds} & \guardok{not followed} & Redirects are relayed, never replayed off-policy. \\
Client auth passthrough & \attackbad{leaks auth} & \guardok{dropped} & Header allowlist; enclave inserts provider credential. \\
Upstream dispatch failure & \attackwarn{host retry} & \guardok{enclave retry} & Failed account excluded before client-visible bytes. \\
Mid-stream failure & \attackbad{replay risk} & \guardok{no replay} & Retry disabled once response bytes are visible. \\
\bottomrule
\end{tabular}
\end{table*}

The remaining checks cover the assumptions behind the harness. For confidentiality, we run the full path and capture the host-side network hop during a live streamed completion. The capture holds only ciphertext, with neither prompt nor response in the clear, and the check rejects an empty capture, which would otherwise pass vacuously. For execution integrity, the platform reports the enclave running with no debugging flags, so the attestation carries real rather than zeroed measurements; the client rejects an empty nonce, zeroed measurements, and any destination outside the enclave-owned policy. For faithfulness, a conformance suite drives the relay with inputs built to trip a naive transformer, including multimodal content, tool-call structures, a structured-output request, exotic field shapes, and a large body, and confirms request and response are carried byte for byte while the client's authorization header is dropped.

\subsection{Adaptive and Coverage-Oriented Evaluation}
A second layer of checks asks how much room the untrusted host still has to influence the protected path, beyond the named attacks of Table~\ref{tab:harness}. The first test enumerates the host/enclave control-channel schemas from source and guards them in unit tests. The authorize request carries only credentials and labels, never a body, header, destination, or response byte, and the reply only allows or denies, selects an account, names an enclave-owned policy, and supplies a provider credential. With three valid provider-policy pairs in the current table, the allow-path destination choice is bounded by $\log_2 3 = 1.58$ bits per request, and those bits select among audited destinations. A separate test sets the host's metadata model to a value different from the model inside the client body and confirms the body sent onward is unchanged. The host therefore retains availability, account, credential, provider, and telemetry influence, but not a body-derived control channel.

The second test targets the check-to-use gap. The sidecar fails closed before constructing a proxy on any verification error, and after a successful check it pins the attested leaf certificate. A synthetic certificate-rotation test confirms that a different certificate fails during the TLS handshake before the rotated endpoint handles any request. The implemented contract is one verified sidecar session, not continuous re-attestation; on an enclave restart with a new certificate, the client establishes a new verified session.

The third test measures the stream side channel rather than hiding it. A fuzz target compares both relay modes, line-at-a-time for server-sent events and whole-body otherwise, against the identity relation over generated JSON bodies and event chunks, covering tool-call structures, multimodal-shaped payloads, unusual encodings, oversized inputs, and chunk boundaries around JSON delimiters. \sys protects body contents but does not pad event sizes or add timing jitter, so a host observing ciphertext record sizes and timing can learn the event cadence and approximate line sizes. We treat this as an explicit residual leakage channel, not as a violation of the body-confidentiality claim.

A machine-checked symbolic model cross-checks the protocol ordering under ideal cryptography. Built in the applied pi calculus and checked with ProVerif, it confirms attested release, check-to-use binding, destination authority, response provenance, and a body-free host view against an adversary that controls every channel, and a falsification suite that removes one protocol check at a time confirms each check is load-bearing (supplementary material). The three tests above supply the facts the model assumes, with the schema tests anchoring the body-free control channel and the conformance and fuzz suites anchoring byte-for-byte faithfulness.

\subsection{Auditability}
The image builds reproducibly~\cite{awsnitrorepro} to a signed measurement the sidecar pins, the binding the protocol of \S\ref{sec:arpa} relies on. The operator signs the build manifest and publishes it to a public append-only transparency log~\cite{sigstore}, and the sidecar pins a measurement only after it verifies the log's inclusion proof, its signed tree head, and a consistency proof against the tree head it last saw. This deploys the accountability of Lemma~\ref{lem:account} rather than assuming it. An operator cannot pin for one client a measurement it hides from the public log or shows differently to another. Against a live public log a genuine manifest verifies end to end, while an unlogged manifest, a manifest that mismatches its logged entry, and a forked or rolled-back log each fail closed before any pin; the entry and verification logs ship with the artifact.

The audit target is the trusted relay, $851$ lines across five files that link no package from the router's service layer. Under a broader source-line count applied identically to both systems, \sys measures $1{,}236$ lines against $10{,}434$ for the nearest open enclave gateway~\cite{opengradient} at revision \texttt{afa7966}, which runs a transforming gateway inside the enclave. This is a source audit, not a running differential. To make the auditability claim testable, we seed ten isolated invariant violations. Codex~\cite{codexcli} found $8/10$, missing a destination-host/path mismatch and a host-controlled model-in-body case; Claude Code~\cite{claudecode} found $10/10$ on the same minimal fixture. Neither flagged the clean implementation. We treat these pilots as calibration evidence that a small faithful relay is reviewable by commodity coding agents, not as proof of human audit completeness. The toolkit, base image, and toolchain are shared and pinned, so the community audits them once.

\begin{figure*}[t]
\centering
\includegraphics[width=0.96\textwidth]{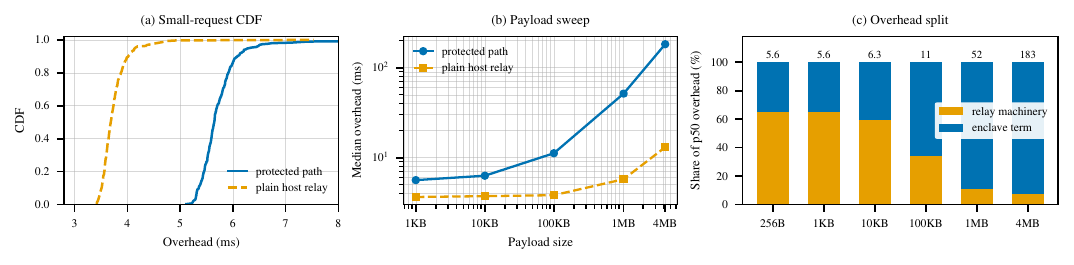}
\caption{Steady-state relay overhead against a local upstream, for the protected path and for the identical relay running as a plain host process. (a) Overhead for a small body, medians $5.7$ and $3.7$~ms. (b) Median overhead against body size; the curve gap is the enclave term. (c) Its decomposition; the enclave term overtakes the relay machinery between $10$ and $100$~KB.}
\label{fig:runtime}
\end{figure*}

\subsection{Workload and Runtime Cost}
The relay adds a few milliseconds per request in steady state, small against a model call of about a second. To measure the relay's own cost without provider noise, we replace the provider with a zero-latency local upstream and reuse pooled connections, the common case for a router that keeps upstream sessions open. For a small request body the added latency has a median of $5.7$~ms and a 95th percentile of $6.3$~ms (Figure~\ref{fig:runtime}a). It stays near this floor up to about $10$~KB. Beyond that it grows with the body, since every byte must be encrypted and carried across the extra hops, reaching $52$~ms at a $1$~MB body and $183$~ms at a $4$~MB multimodal-scale body (Figure~\ref{fig:runtime}b).

A plain-router control isolates the boundary's own cost. Running the identical relay logic as an ordinary host process costs $3.7$~ms at a small body (Figure~\ref{fig:runtime}a, dashed), so the confidentiality boundary adds $2.0$~ms over the plaintext router it replaces, covering the virtual-socket data path, the dedicated $2$-vCPU allocation, and the verifying sidecar hop. As the body grows the gap widens (Figure~\ref{fig:runtime}b) and the split inverts (Figure~\ref{fig:runtime}c), with the enclave term overtaking the relay machinery between $10$ and $100$~KB and dominating at multimodal scale. Large bodies are expensive because bytes must be carried and encrypted across the boundary, not because the relay logic is slow. Opening an attested session takes a further $12$~ms, paid once per session.

\begin{figure*}[t]
\centering
\includegraphics[width=\textwidth]{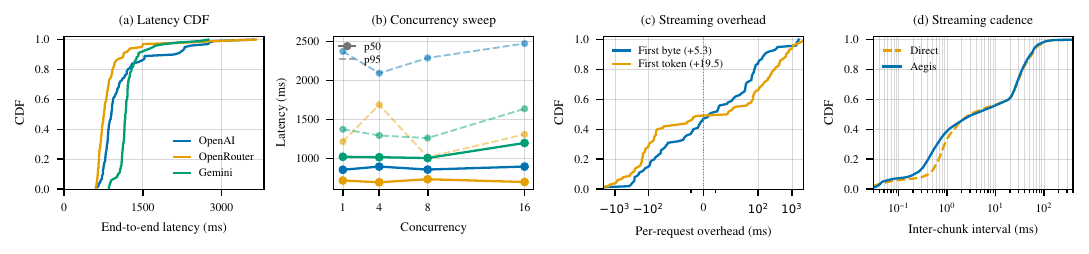}
\caption{Real-provider workloads and streaming through the verified path. (a) End-to-end latency CDFs, $100$ provider-native requests per provider, medians $902$ to $1175$~ms. (b) Median (solid) and 95th-percentile (dashed) latency against concurrency; all $600$ requests succeed, each median within $20\%$ of its single-request value. (c) Paired per-request overhead on a fixed $256$-token streamed completion, median $5.3$~ms to the first byte and $19.5$~ms to the first token. (d) Inter-chunk intervals through \sys and direct coincide (medians $3.5$ and $3.0$~ms).}
\label{fig:providers}
\end{figure*}

The remaining measurements use real providers. \sys carries three provider-native endpoints with no in-enclave conversion layer, OpenAI Responses, OpenRouter Chat Completions, and Gemini generateContent, and exercises each through the verified path. A direct-provider runner that only validates keys isolates provider variance, and a \sys-through runner sends synthetic requests through a verified sidecar; only the latter supports end-to-end claims. All $100$ \sys-through requests per provider succeeded and returned the expected marker, and end-to-end latency is dominated by provider variance rather than by the relay (Figure~\ref{fig:providers}a).

Streaming through the boundary preserves the provider's token cadence. Over $100$ streamed completions through the verified path and $100$ sent directly to the same provider, interleaved request by request, the inter-chunk arrival intervals coincide (Figure~\ref{fig:providers}d), so the enclave hop adds no measurable jitter between tokens. On those paired runs the boundary adds a median of $5.3$~ms to the first byte, matching the steady-state relay cost, and $19.5$~ms to the first token (Figure~\ref{fig:providers}c), while the per-request first-token latency is set by the provider and varies by hundreds of milliseconds.

The $2$-vCPU enclave is not a bottleneck at moderate parallelism. In the three-provider concurrency run all $600$ requests succeeded and median latency stays within $20\%$ of its single-request value through concurrency $16$ (Figure~\ref{fig:providers}b).

In short, confidentiality and integrity cost a small constant on small requests and a bounded copy cost on large ones, well under the latency of the model call they protect.
 \section{Discussion and Limitations}
\label{sec:discussion}

\heading{An unavoidable trust anchor}
When the upstream is an ordinary first-party provider that speaks plaintext, some party between the client and that provider must terminate the client's session and hold the body in the clear, so a trust anchor on that relay code is unavoidable. Only its form is open, not its existence. Ours today is a per-release human audit of a minimal relay, with its measurement anchored in a public append-only log; an accepted measurement is then one the operator has publicly committed to (\S\ref{sec:arpa}). The audit can still move toward a machine-checked non-interference property, which remains future work. Undecidability rules out a complete automatic check of arbitrary code, so a plaintext-holding relay always rests on reading its code or machine-checking a restricted property of it.

\heading{Removing plaintext access}
The relay holds plaintext only because a first-party upstream cannot. When the next hop is another attested router, the relay can require attested transport from it and forward only to an allowlisted measurement, the attested-upstream direction of \S\ref{sec:design}. It then becomes a blind ciphertext forwarder that needs no body-confidentiality anchor, and trust extends per attested hop rather than resting on the first relay. This path eliminates the plaintext anchor, and our single-hop prototype does not yet take it.

\heading{Residual leakage}
\sys protects body contents only up to the explicit leakage profile of \S\ref{sec:goals}, the sizes, timing, cadence, account, provider, and usage the host still observes. This residual is exploitable. An observer who records only ciphertext record sizes and timing can recover the per-token length sequence of a streamed response and reconstruct text from it~\cite{promptkeylog}. For a remote eavesdropper the obstacles are reaching an on-path vantage and training a per-target reconstruction model. The router operator faces neither. It terminates every connection, reads record sizes from the buffers it forwards, and its upstream accounts supply known-plaintext training data. The relay also re-frames each response toward the client, so transport-level padding or batching applied upstream does not survive the hop.

The enclave is the one place on the path where a defense sits outside the operator's control. Traffic shaping baked into the measured image is attested and the host cannot disable it, a guarantee no plaintext gateway can make. Record-level padding would close the size dimension while preserving byte-for-byte faithfulness, and timing defenses fit at a larger trusted base. We leave this defense and its trusted-base cost to future work.

\heading{Boundaries}
The enclave checks a real provider's certificate in the usual way, against public certificate authorities and the expected hostname. It does not pin the provider's specific key. A compromised authority could therefore issue a fake certificate that the enclave accepts. Pinning would close this gap, but routine key rotation would then break requests until the enclave image is rebuilt and its measurement is republished. We leave this as follow-on hardening.

The transparency-log check rests on a trusted external anchor in the same way. The sidecar verifies an inclusion proof and a signed tree head at each pin and checks consistency against the head it last saw, which rules out a log that drops or rewrites the operator's entry across one client's repeated observations. A log that shows different clients different views is ruled out by the external witnesses and gossip that monitor a public log, not by any single client's check.

Egress confinement is platform-dependent. The relay resolves every destination inside its measured image, so a compromised host cannot redirect plaintext at the application layer, but forcing that confinement below the image is a platform property. On the current platform the host process performs the enclave's outbound networking, which makes below-image egress control hard to enforce. We leave it out of the present prototype.

Because \sys keeps no security-relevant state in the enclave across requests, the rollback and forking attacks that target stateful enclaves do not apply. The host can still roll back its own usage records, which we leave unprotected; a deployment that needs tamper-proof accounting can add state continuity on top.

A new enclave version needs a new check before the client trusts it, but the split keeps this cheap. Provider, billing, scheduling, authentication, and management changes stay on the host and leave the measurement untouched; only changes to the small relay core require a fresh check (\S\ref{sec:arpa}). A client that does not rebuild locally can accept a new measurement signed by a key it already trusts.

\heading{Comparison}
\S\ref{sec:related} compares \sys with the closest open gateway in detail. The design difference that matters here is the integrity anchor. Ours rides the live attested session, fresh and non-replayable by construction, while the open gateway offers portable signed receipts and network-level unlinkability that we do not target. A deployment that wants offline checking can add receipts bound to fresh client nonces.

A further extension would encrypt provider keys under an enclave-held key, letting one router serve mutually distrusting tenants, an isolation our single-operator design does not need.
 \section{Related Work}
\label{sec:related}

\sys spans attested LLM gateways, client-side defenses, enclave-shielded systems, and attested deployment. Table~\ref{tab:related} summarizes the comparison.

\heading{The closest prior art}
An open-source enclave gateway~\cite{opengradient} is our nearest neighbor. It routes to third-party providers from a cloud enclave, and attests its image and signs response hashes. Its posture is verify-after-receive, while \sys refuses before sending. The gateway also runs a transforming router inside the enclave, with provider selection and payment logic. \sys keeps only a faithful passthrough there, so the audited measurement stays small. The gateway covers more providers and offers network-level unlinkability and portable receipts.

Portcullis~\cite{portcullis} is the closest peer-reviewed system. It rewrites prompts to mask sensitive fields and restores them on return, mainly protecting privacy against the provider, and lighter-weight sanitizers share the posture~\cite{prosan}. \sys passes content through faithfully because rewriting would break tool-call semantics. \sys protects integrity and confidentiality against the router operator.

\heading{Guardrails and signed responses}
Agent-side systems screen prompts and tool calls at the client, including guardrail frameworks~\cite{llamafirewall} and control-flow integrity for agents~\cite{camel}. They observe only what the router delivers, so a schema-valid rewrite passes and request-path secrets stay exposed. Signature-based provenance for LLM APIs verifies responses after receipt and needs the serving side to cooperate~\cite{aex}. \sys removes the rewrite point instead.

\begin{table*}[!t]
\centering
\caption{Comparison with prior approaches against a malicious LLM API router.}
\label{tab:related}
\footnotesize
\renewcommand{\arraystretch}{1.16}
\newcommand{\pfull}{\tikz[baseline=-0.55ex]\fill[green!75!black] (0,0) circle (0.55ex);}
\newcommand{\phalf}{\tikz[baseline=-0.55ex]{\fill[orange!85!black] (0,0) -- (90:0.55ex) arc (90:270:0.55ex) -- cycle;\draw[orange!85!black,line width=0.5pt] (0,0) circle (0.55ex);}}
\newcommand{\pnone}{\tikz[baseline=-0.55ex]\draw[red!70!black,line width=0.5pt] (0,0) circle (0.55ex);}
\newcommand{\pna}{--}
\newcommand{\pq}[1]{\,{\scriptsize #1}}
\resizebox{\textwidth}{!}{%
\begin{tabular}{@{}llllll@{}}
\toprule
\multicolumn{1}{@{}c}{\textbf{Approach}} & \multicolumn{1}{c}{\textbf{Trusted component}} & \multicolumn{1}{c}{\textbf{Client verification}} & \multicolumn{1}{c}{\textbf{Faithful relay}} & \multicolumn{1}{c}{\textbf{Anchored measurement}} & \multicolumn{1}{c@{}}{\textbf{Attack coverage}} \\
\midrule
Plaintext routers~\cite{litellm,openrouter,portkey,kongai} & operator honesty & \pnone\pq{none} & \pnone\pq{may rewrite} & \pna & \pnone\pq{all four succeed} \\
Client-side defenses~\cite{yaim,llamafirewall,camel} & client-side filters & \phalf\pq{after receive} & \pnone\pq{unenforced} & \pna & \pnone\pq{evadable} \\
Provider-signed responses~\cite{yaim,aex} & provider signatures & \phalf\pq{after receive} & \phalf\pq{response only} & \pna & \phalf\pq{provider opt-in} \\
Prompt-masking proxies~\cite{portcullis,prosan} & masking proxy\textsuperscript{a} & \pnone\pq{none} & \pnone\pq{rewrites by design} & \pna & \pnone\pq{targets the provider} \\
Attested gateways~\cite{opengradient,redpill,phala} & enclave, transforming router & \phalf\pq{after receive} & \pnone\pq{transforms} & \phalf\pq{image only} & \phalf\pq{after the fact} \\
\midrule
Attested model serving~\cite{applepcc,tinfoil,privatemode} & enclave, whole serving stack & \pfull\pq{before release} & \pna\pq{serves own models} & \phalf\pq{binaries, log} & \pfull\pq{forfeits provider choice} \\
Direct provider access & provider only & \pfull\pq{provider TLS} & \pfull\pq{no intermediary} & \pna & \pfull\pq{forfeits routing} \\
\midrule
\multicolumn{1}{c}{\textbf{\sys}} & \textbf{enclave, faithful passthrough} & \pfull\,\textbf{before release} & \pfull\,\textbf{byte-for-byte} & \pfull\,\textbf{audit, rebuild, log} & \pfull\,\textbf{all four, routing kept} \\
\bottomrule
\addlinespace[2pt]
\multicolumn{6}{@{}l@{}}{\scriptsize \pfull~yes \quad \phalf~partial \quad \pnone~no \quad \pna~not applicable \qquad \textsuperscript{a}\,Portcullis attests its masking gateway; ProSan runs as a local sanitizer.} \\
\end{tabular}%
}
\end{table*}

\heading{Enclave-shielded systems}
Prior work shields network intermediaries inside enclaves, including cloud network functions~\cite{safebricks} and anonymity relays~\cite{sgxtor}. These systems show that an enclave can protect an intermediary. \sys applies the same idea to LLM application traffic.

Confidential-AI work protects the model pipeline, including confidential LLM serving~\cite{pipellm} and verifiable private execution~\cite{slalom}. Production systems deploy the same posture for assistant features~\cite{whatsapppp}, and attested model serving lets a client verify the serving stack before releasing a prompt, anchored to published binaries and a transparency log~\cite{applepcc} or reproducibly built images~\cite{tinfoil,privatemode}; the client adopts the hosted models. These are first-party or platform-controlled settings. \sys places no model in the enclave and targets a third-party router that relays to the client's chosen provider.

Shielding only part of a model is not automatically safe. Attacks recover models or inputs from partitioned deployments~\cite{noprivacy,gameofarrows}. The enclave in \sys holds only a faithful-passthrough data plane, with no weight or tensor to attack.

\heading{Attestation}
This part of \sys builds on attested deployment and supply-chain transparency. We use an enclave toolkit~\cite{nitriding} and platform support for in-enclave TLS termination, certificate binding, attestation, and reproducible measurements~\cite{awsnitrotls,awsnitroattest,awsnitrorepro}. We draw on attestation bound into transport security~\cite{ratls}, software signing with transparency logs~\cite{sigstore}, and enclave-built verifiable binaries~\cite{attestablebuilds}.

Remote-attestation work covers time-of-check to time-of-use gaps~\cite{toctoura}, and state-continuity systems address rollback and forked state~\cite{memoir}. \sys avoids that problem by holding no cross-request state in the enclave.

Hardware trust is conditional. Prior work breaks enclave roots through side channels~\cite{controlledchannel} and transient execution~\cite{foreshadow}. Confidential-virtual-machine attacks~\cite{cachewarp} target other platforms; attestation is not proof of absolute security.
 \section{Conclusion}
\label{sec:conclusion}

Malicious LLM API routers are dangerous. Plaintext at the TLS-terminating router makes tool-call rewrites, trigger-gated attacks, and passive secret exfiltration possible. We propose \sys, an attested faithful-passthrough router that makes the enclave the only component touching plaintext.  We empirically demonstrate that \sys blocks all four malicious-router attacks while preserving provider-native routing with low overhead. 

\appendices

\section{Proofs}
\label{app:proofs}

We give each reduction as an explicit sequence of games on a common probability space, following the methodology of Shoup~\cite{shoupgames}. Write $S_i$ for the event that Game~$i$ returns the value favorable to the adversary, the guess $\beta'{=}\beta$ for the confidentiality experiment and the bad event for the integrity, authority, and accountability experiments. Each step is one of three kinds: a \emph{bridging} step that rewrites the game with $\Pr[S_i]$ unchanged; a \emph{failure-event} step, where two games coincide unless an event $F$ occurs and $|\Pr[S_i]-\Pr[S_{i+1}]|\le\Pr[F]$ by the difference lemma; or an \emph{indistinguishability} step bounded by a primitive's distinguishing advantage. The advantage symbols are those of \S\ref{sec:prelim}.

\newcommand{\gamebox}[1]{\fbox{\parbox{\dimexpr\linewidth-2\fboxsep-2\fboxrule}{\begin{tabular}{@{}l@{}}#1\end{tabular}}}}
\begin{figure*}[t]
\small
\centering
\begin{minipage}[t]{0.48\textwidth}
\gamebox{%
\textbf{Game } $\mathsf{G}^{\mathrm{conf}}_{\scheme,\mathcal{A}}(\lambda)$: \\[2pt]
1.\ $\mathit{pp}\gets\HW.\mathsf{Setup}(1^\lambda)$ \\
2.\ $(b_0,b_1,\mathit{st})\gets\mathcal{A}(\mathit{pp})$;\ \ $\leak(b_0){=}\leak(b_1)$ \\
3.\ $\beta\xleftarrow{\$}\{0,1\}$ \\
4.\ $\beta'\gets\mathcal{A}^{\mathsf{Rel}(b_\beta)}(\mathit{st})$ \\
5.\ \textbf{return } $[\beta'{=}\beta]$
}
\end{minipage}
\hfill
\begin{minipage}[t]{0.48\textwidth}
\gamebox{%
\textbf{Game } $\mathsf{G}^{\mathrm{int}}_{\scheme,\mathcal{A}}(\lambda)$: \\[2pt]
1.\ $\mathit{pp}\gets\HW.\mathsf{Setup}(1^\lambda)$ \\
2.\ $(b,\mathit{st})\gets\mathcal{A}(\mathit{pp})$ \\
3.\ $(\hat b,\hat r,r)\gets\mathsf{Sess}_{\mathcal{A}}(b)$ \\
4.\ \textbf{return } $[\hat b{\neq}b\ \vee\ \hat r{\neq}r]$
}
\end{minipage}

\vspace{0.7em}

\begin{minipage}[t]{0.48\textwidth}
\gamebox{%
\textbf{Game } $\mathsf{G}^{\mathrm{str}}_{\scheme,\mathcal{A}}(\lambda)$: \\[2pt]
1.\ $\mathit{pp}\gets\HW.\mathsf{Setup}(1^\lambda)$ \\
2.\ $(b,\mathit{st})\gets\mathcal{A}(\mathit{pp})$ \\
3.\ $(\mathit{acc},\mathit{eos})\gets\mathsf{Sess}_{\mathcal{A}}(b)$ \\
4.\ \textbf{return } $[\mathit{acc}{=}\textsf{done}\ \wedge\ \neg\mathit{eos}]$
}
\end{minipage}
\hfill
\begin{minipage}[t]{0.48\textwidth}
\gamebox{%
\textbf{Game } $\mathsf{G}^{\mathrm{dest}}_{\scheme,\mathcal{A}}(\lambda)$: \\[2pt]
1.\ $\mathit{pp}\gets\HW.\mathsf{Setup}(1^\lambda)$ \\
2.\ $(b,\mathit{st})\gets\mathcal{A}(\mathit{pp})$ \\
3.\ $d\gets\mathsf{Sess}_{\mathcal{A}}(b)$ \\
4.\ \textbf{return } $[d\notin\Pol]$
}
\end{minipage}

\vspace{0.7em}

\begin{minipage}[t]{0.48\textwidth}
\gamebox{%
\textbf{Game } $\mathsf{G}^{\mathrm{anch}}_{\mathcal{A}}(\lambda)$: \\[2pt]
1.\ $\mathit{pp}\gets\HW.\mathsf{Setup}(1^\lambda)$ \\
2.\ $\mathsf{Audit},\mathsf{Pin}(s)$;\ \ $m^{\ast}{\gets}\mathsf{Build}(s,R)$ \\
3.\ $\mathit{hdl}\gets\mathsf{Sess}_{\mathcal{A}}$ \\
4.\ \textbf{return } $[\mathsf{img}(\mathit{hdl})\neq\mathsf{img}_{s}]$
}
\end{minipage}
\hfill
\begin{minipage}[t]{0.48\textwidth}
\gamebox{%
\textbf{Game } $\mathsf{G}^{\mathrm{acc}}_{\mathcal{A}}(\lambda)$: \\[2pt]
1.\ $\mathit{pp}\gets\HW.\mathsf{Setup}(1^\lambda)$ \\
2.\ $(\mathit{mf},\sigma,\pi,\eta,m)\gets\mathcal{A}(\mathit{pp})$ \\
3.\ \textbf{if } $\neg\mathsf{Vrfy}_{\mathrm{acc}}(\mathit{mf},\sigma,\pi,\eta,m)$:\ \textbf{ret } $0$ \\
4.\ \textbf{return } $[m\text{ unbound to signed/logged }s]$
}
\end{minipage}
\caption{The six experiments as games. The adversary $\mathcal{A}$ is the host; $\mathsf{Sess}_{\mathcal{A}}(b)$ and $\mathsf{Rel}(b_\beta)$ are the oracles of Definitions~\ref{def:sess} and~\ref{def:rel}, $\mathsf{img}_{s}$ the image reproducibly built from source $s$, and $\mathsf{Vrfy}_{\mathrm{acc}}$ the honest accountability check.}
\label{fig:games}
\end{figure*}

Figure~\ref{fig:games} renders the six experiments as games. Each proof below starts from the corresponding box as its Game~0 and transforms it by the steps the headings name.

\begin{definition}[Verified-session oracle]
\label{def:sess}
$\mathsf{Sess}_{\mathcal{A}}(b)$ runs one protocol session on body $b$ with $\mathcal{A}$ controlling the host. The adversary chooses the image to load, obtaining a handle $\mathit{hdl}$, and drives every host-side channel, the account selection, and the host control channel. The honest sidecar runs $\mathsf{Attest},\mathsf{Vrfy}$ and releases the encrypted body to the enclave only when $\mathsf{Vrfy}$ accepts a quote binding the pinned measurement $m^{\ast}$, the fresh nonce $n$, and the channel key $k$, and otherwise releases $\bot$. The honest provider at the resolved destination returns response $r$. The oracle exposes the body $\hat b$ the provider receives, the response $\hat r$ the client accepts, the destination $d$ a plaintext body reaches, the completion flag $\mathit{acc}$, the provider end-of-stream flag $\mathit{eos}$, and the handle $\mathit{hdl}$, and the entire host view of the session.
\end{definition}

\begin{definition}[Release oracle]
\label{def:rel}
$\mathsf{Rel}(b_\beta)$ is $\mathsf{Sess}_{\mathcal{A}}$ specialized to the confidentiality challenge. It releases $\mathsf{Enc}_k(b_\beta)$ to the enclave only on a quote that $\mathsf{Vrfy}$ accepts, and returns to $\mathcal{A}$ the entire host view of the resulting session, namely the channel ciphertexts, the host control-channel fields, and the running enclave's host-observable memory. The adversary queries $\mathsf{Rel}$ and outputs a guess $\beta'$.
\end{definition}

The games separate what each reduction proves from what it assumes. The confidentiality, integrity, and destination games take the loaded program's faithfulness as an assumption on $Q$, its host non-interference $\NonInt(Q,\leak)$, verbatim relay $\Verb(Q)$, and baked-destination resolution $\DestPol(Q,\Pol)$, rather than as a conclusion. Lemma~\ref{lem:anchor} discharges that $Q$ is the audited program, and the audit, the conformance and schema tests of \S\ref{sec:eval}, and the symbolic model of Appendix~\ref{app:symbolic} are the present evidence that the audited program is faithful, with full machine-checked non-interference left to future work.

\heading{Confidentiality (Theorem~\ref{thm:conf})}
\textit{Game 0} is the equal-leakage confidentiality experiment $\mathsf{G}^{\mathrm{conf}}$ of Figure~\ref{fig:games}. The challenger samples $\beta\xleftarrow{\$}\{0,1\}$ and runs $\scheme$ on $b_\beta$ against the $\mathcal{A}$-controlled host, which chooses the image loaded into the enclave; the sidecar releases the encrypted body only after attest-and-verify binds the pinned measurement $m^{\ast}$, the fresh nonce, and the session's channel key. Then $\Adv^{\mathrm{conf}}_{\mathcal{A},\scheme}=|\Pr[S_0]-\tfrac12|$.

\textit{Game 1 (failure event, RA-unforgeability).} Abort and output a random bit if the sidecar accepts evidence on a measurement-input-output tuple no enclave produced. Games 0 and 1 coincide unless this event $F_1$ occurs, and $\Pr[F_1]\le\Adv^{\RAunf}_{\mathcal{A},\HW}$, so $|\Pr[S_0]-\Pr[S_1]|\le\Adv^{\RAunf}_{\mathcal{A},\HW}$. Hereafter an accepted session runs an enclave whose loaded image measures $m^{\ast}$, the pinned data path $Q$, which the theorem assumes satisfies $\NonInt(Q,\leak)$ (Lemma~\ref{lem:anchor} ties $m^{\ast}$ to the audited source, a separate result whose terms are not summed here).

\textit{Game 2 (failure event, execution integrity).} Abort if the running enclave produces an accepted output or transcript other than that of the honest $Q$, that is, the host drives $Q$ off its semantics. Then $\Pr[F_2]\le\Adv^{\ExeInty}_{\mathcal{A},\HW}$, so $|\Pr[S_1]-\Pr[S_2]|\le\Adv^{\ExeInty}_{\mathcal{A},\HW}$. Hereafter every host-facing output of $Q$ is the honest one.

\textit{Game 3 (bridging, non-interference).} Rewrite those outputs as $\NonInt(Q,\leak)$ specifies, a function of $\leak$ alone. By the hypothesis $\NonInt(Q,\leak)$ this is the same distribution, so $\Pr[S_3]=\Pr[S_2]$; the step uses no computational assumption.

\textit{Game 4 (indistinguishability, client channel).} Idealize the client-to-enclave secure channel, replacing the host's view of every message in the session, the request ciphertext up and the response ciphertext down, by the ideal-channel transcript that reveals only message lengths and timing; the enclave still receives the real body. A distinguisher breaks that secure channel, so $|\Pr[S_3]-\Pr[S_4]|\le\Adv^{\mathrm{ch}_{\mathrm{client}}}_{\mathcal{A}}$.

\textit{Game 5 (indistinguishability, provider channel).} Idealize the enclave-to-provider secure channel the same way, in both directions, so the host's view of the forwarded body and the returned response reduces to lengths and timing; $|\Pr[S_4]-\Pr[S_5]|\le\Adv^{\mathrm{ch}_{\mathrm{prov}}}_{\mathcal{A}}$.

\textit{Game 6 (indistinguishability, memory confidentiality).} After Games 4 and 5 the host's only $\beta$-dependent view is the enclave's running memory, which holds the body and the provider's reply. Run the enclave through the isolation challenger on the admissible pair $b_0,b_1$, the provider acting as the fixed oracle of the leakage profile so its observable reply is equal across the pair, and swap the secret from $b_\beta$ to the fixed $b_0$; this is a left-or-right isolation step whose distinguisher is an isolation distinguisher, so $|\Pr[S_5]-\Pr[S_6]|\le\Adv^{\Iso}_{\mathcal{A},\HW}$. The reduction holds no plaintext $b_\beta$, since every wire view is now the length-only ideal-channel transcript and the challenger runs the enclave on the secret.

In Game 6 nothing depends on $\beta$, so $\Pr[S_6]=\tfrac12$. Summing the six transitions gives the bound of Theorem~\ref{thm:conf}. \qed

\heading{Faithful relay and integrity (Theorem~\ref{thm:faithful})}
\textit{Game 0} is the faithful-relay experiment $\mathsf{G}^{\mathrm{int}}$ on a submitted body $b$, with $r$ the response the chosen provider returns; $S_0$ is the bad event that the provider receives some $b'\neq b$ or the client accepts some $r'\neq r$, and $\Adv^{\mathrm{int}}_{\mathcal{A},\scheme}=\Pr[S_0]$.

\textit{Game 1 (failure event, execution integrity).} Abort if the host drives $Q$ off the verbatim-relay semantics $\Verb(Q)$, so $|\Pr[S_0]-\Pr[S_1]|\le\Adv^{\ExeInty}_{\mathcal{A},\HW}$. Hereafter $Q$ forwards the request and relays the response, including every chunk through the end-of-stream marker, byte for byte, emitting the marker to the client only when it relays one from the provider.

\textit{Game 2 (failure event, enclave-to-provider channel).} Abort if any message on the enclave-to-provider channel is altered in either direction, the request body delivered to the provider differing from the one $Q$ sent, or a response chunk or end-of-stream marker $Q$ receives differing from the one the provider sent; both break that secure channel, so $|\Pr[S_1]-\Pr[S_2]|\le\Adv^{\mathrm{ch}_{\mathrm{prov}}}_{\mathcal{A}}$.

\textit{Game 3 (failure event, client-to-enclave channel).} Abort if the response the client accepts differs from the one $Q$ returned, or the client accepts as completed a stream whose marker $Q$ never relayed; both break the client-to-enclave secure channel, so $|\Pr[S_2]-\Pr[S_3]|\le\Adv^{\mathrm{ch}_{\mathrm{client}}}_{\mathcal{A}}$.

In Game 3 the provider receives exactly $b$, the client accepts exactly $r$, and a stream lacking the provider's marker is treated as incomplete, so $\Pr[S_3]=0$ and $\Adv^{\mathrm{int}}_{\mathcal{A},\scheme}\le\Adv^{\ExeInty}_{\mathcal{A},\HW}+\Adv^{\mathrm{ch}_{\mathrm{client}}}_{\mathcal{A}}+\Adv^{\mathrm{ch}_{\mathrm{prov}}}_{\mathcal{A}}$. The same three transitions rule out the streaming-integrity bad event $\mathsf{G}^{\mathrm{str}}$, a client accepting a truncated response as completed: a false completion needs a forged end-of-stream marker, emitted by a deviating $Q$ (Game 1), forged on the provider channel and relayed (Game 2), or forged on the client channel (Game 3). Hence $\Adv^{\mathrm{str}}_{\mathcal{A},\scheme}\le\Adv^{\ExeInty}_{\mathcal{A},\HW}+\Adv^{\mathrm{ch}_{\mathrm{client}}}_{\mathcal{A}}+\Adv^{\mathrm{ch}_{\mathrm{prov}}}_{\mathcal{A}}$. \qed

\heading{Destination authority (Theorem~\ref{thm:dest})}
\textit{Game 0} is the destination-authority experiment $\mathsf{G}^{\mathrm{dest}}$; $S_0$ is the bad event that a plaintext body reaches some $d\notin\Pol$, and $\Adv^{\mathrm{dest}}_{\mathcal{A},\scheme}=\Pr[S_0]$.

\textit{Game 1 (failure event, execution integrity).} Abort if the host drives $Q$ off $\DestPol(Q,\Pol)$, so that $Q$ resolves a destination outside the baked set; then $|\Pr[S_0]-\Pr[S_1]|\le\Adv^{\ExeInty}_{\mathcal{A},\HW}$. Hereafter every session $Q$ opens names a hostname in $\Pol$.

\textit{Game 2 (failure event, endpoint-authenticated provider channel).} Abort if plaintext reaches a peer that is not the authenticated endpoint for the resolved hostname; redirecting the session to a different peer breaks the channel's endpoint authentication, so $|\Pr[S_1]-\Pr[S_2]|\le\Adv^{\mathrm{ch}_{\mathrm{prov}}}_{\mathcal{A}}$.

In Game 2 plaintext reaches only the certificate-authenticated endpoint for a hostname in $\Pol$, so $\Pr[S_2]=0$. The sole remaining strategy, a certificate authority that mis-issues for a pinned hostname, is excluded by the trusted public-key-infrastructure assumption of \S\ref{sec:prelim} and contributes no term. Hence $\Adv^{\mathrm{dest}}_{\mathcal{A},\scheme}\le\Adv^{\ExeInty}_{\mathcal{A},\HW}+\Adv^{\mathrm{ch}_{\mathrm{prov}}}_{\mathcal{A}}$. \qed

\heading{Measurement anchoring (Lemma~\ref{lem:anchor})}
\textit{Game 0} ($\mathsf{G}^{\mathrm{anch}}$) runs the audit-and-pin phase over source $s$ with $m^{\ast}=\mathsf{Build}(s,R)$, then a session in which the sidecar releases plaintext; $S_0$ is the bad event that the release reaches an enclave not running the image built from $s$.

\textit{Game 1 (failure event, RA-unforgeability and platform soundness).} Abort if the sidecar accepts evidence not produced by a genuine enclave session running an image measuring $m^{\ast}$. Since the platform root signs only for genuinely loaded images, $|\Pr[S_0]-\Pr[S_1]|\le\Adv^{\RAunf}_{\mathcal{A},\HW}$, and hereafter the accepting session runs an image measuring $m^{\ast}$.

\textit{Game 2 (failure event, collision resistance).} Abort if that image differs from the one the client rebuilt from $s$ while both measure to $m^{\ast}$, so $|\Pr[S_1]-\Pr[S_2]|\le\Adv^{\mathrm{cr}}_{\mathcal{A}}$.

\textit{Game 3 (failure event, EUF-CMA).} Abort if the client pinned $m^{\ast}$ from an operator signature on source other than $s$, so $|\Pr[S_2]-\Pr[S_3]|\le\Adv^{\mathrm{euf\text{-}cma}}_{\mathcal{A}}$.

In Game 3 the release reaches an enclave running the image built from $s$, and because the sidecar encrypts to the attested channel key the secure channel carries it to that enclave alone, so $\Pr[S_3]=0$. A release thus reaches only the audited image except with probability at most $\Adv^{\RAunf}_{\mathcal{A},\HW}+\Adv^{\mathrm{cr}}_{\mathcal{A}}+\Adv^{\mathrm{euf\text{-}cma}}_{\mathcal{A}}$, the bound of Lemma~\ref{lem:anchor}. \qed

\heading{Accountability (Lemma~\ref{lem:account})}
\textit{Game 0} ($\mathsf{G}^{\mathrm{acc}}$) returns $1$ iff an honest verifier accepts a live measurement not bound to source the operator signed, that the log records, and that reproducibly rebuilds to it; $\Adv^{\mathrm{acc}}_{\mathcal{A}}=\Pr[S_0]$.

\textit{Game 1 (failure event, EUF-CMA).} Abort if the manifest's operator signature verifies on source the operator never signed, so $|\Pr[S_0]-\Pr[S_1]|\le\Adv^{\mathrm{euf\text{-}cma}}_{\mathcal{A}}$.

\textit{Game 2 (failure event, collision resistance).} Abort if two distinct manifests rebuild to the same measurement, so $|\Pr[S_1]-\Pr[S_2]|\le\Adv^{\mathrm{cr}}_{\mathcal{A}}$.

\textit{Game 3 (failure event, log inclusion and consistency).} Abort if the verifier accepts an inclusion proof for a manifest the log does not record, or reads a position whose contents differ from another honest verifier's; both break the log assumption, so $|\Pr[S_2]-\Pr[S_3]|\le\Adv^{\mathrm{log}}_{\mathcal{A}}$.

In Game 3 none of the three events occurs, so acceptance implies the binding and $\Pr[S_3]=0$. Summing gives the bound of Lemma~\ref{lem:account}. \qed

\heading{Security summary}
Table~\ref{tab:summary} collects each result with the advantage terms its bound sums and the structural assumptions its reduction takes on the loaded program $Q$ and the trust roots. The three reduction theorems each consume one clause of the faithfulness predicate, while the two anchoring lemmas consume none.

\begin{table*}[t]
\centering
\caption{Security summary. Each advantage is bounded by the sum of the listed terms of \S\ref{sec:prelim}; the structural assumptions are the non-probabilistic conditions on the loaded program $Q$ and the trust roots.}
\label{tab:summary}
\footnotesize
\renewcommand{\arraystretch}{1.25}
\begin{tabular}{@{}lll@{}}
\toprule
\multicolumn{1}{@{}c}{\textbf{Result}} & \multicolumn{1}{c}{\textbf{Advantage terms}} & \multicolumn{1}{c@{}}{\textbf{Structural assumptions}} \\
\midrule
Confidentiality (Thm~\ref{thm:conf}) & $\Iso,\ \ExeInty,\ \RAunf,\ \mathrm{ch}_{\mathrm{client}},\ \mathrm{ch}_{\mathrm{prov}}$ & $\NonInt(Q,\leak)$ \\
Faithful relay and streaming (Thm~\ref{thm:faithful}) & $\ExeInty,\ \mathrm{ch}_{\mathrm{client}},\ \mathrm{ch}_{\mathrm{prov}}$ & $\Verb(Q)$ \\
Destination authority (Thm~\ref{thm:dest}) & $\ExeInty,\ \mathrm{ch}_{\mathrm{prov}}$ & $\DestPol(Q,\Pol)$, trusted public-key infrastructure \\
Measurement anchoring (Lem~\ref{lem:anchor}) & $\RAunf,\ \mathrm{cr},\ \mathrm{euf\text{-}cma}$ & deterministic build, platform soundness \\
Accountability (Lem~\ref{lem:account}) & $\mathrm{euf\text{-}cma},\ \mathrm{cr},\ \mathrm{log}$ & none \\
\bottomrule
\end{tabular}
\end{table*}

\section{Symbolic Protocol Model}
\label{app:symbolic}

This supporting analysis cross-checks the protocol ordering of the game-based reductions of \S\ref{sec:formal} against a network adversary under ideal cryptography. It machine-checks the abstract protocol state machine, the bootstrapping protocol of \S\ref{sec:arpa} followed by the data path of Figure~\ref{fig:arch}, modeled in the applied pi calculus and verified with ProVerif~\cite{proverif}; the model, the falsification variants, and all verification logs ship with the artifact.

\heading{Model and adversary}
The adversary is the untrusted host. Every channel is a public Dolev--Yao channel, including the sidecar-to-enclave network, the host control channel, and the enclave-to-provider network, so the adversary reads, injects, replays, drops, and reorders all traffic, speaks the control protocol, and runs client sessions of its own. Encrypted payloads are protected by ideal public-key and authenticated encryption, so a public channel means adversary-controlled transport, not plaintext transport; lengths, timing, and traffic shape are not represented. The adversary holds a platform oracle that signs attestation evidence for any measurement other than the pinned $m^{\ast}$, modeling a host free to load and attest arbitrary images, and it chooses the policy name and the provider credential on every control reply, since account and policy selection are host authority. The honest parties are the sidecar holding the body and the pin, the enclave whose every session generates an ephemeral channel key bound into its quote, and the providers at the destinations in $\Pol$.

\heading{Event vocabulary}
The model marks an event at each step the properties constrain, placed after the checks that step performs, so each event asserts that its checks succeeded. $\mathsf{Qt}(m,k,n)$ records that the platform attests measurement $m$ inside an enclave session, binding the session's channel key $k$ and the nonce $n$. $\mathsf{Vrf}(m,k,n)$ records that the sidecar accepts evidence, after the signature, measurement, nonce, and channel-key checks. $\mathsf{Rel}(b,k,n,q)$ records that the sidecar releases body $b$ toward $k$ under request identifier $q$, and $\mathsf{Acc}(b,k,n,q)$ that an enclave session accepts $b$ after decryption under its own key. $\mathsf{Res}(p,d)$ records that the enclave resolves policy name $p$ to destination $d$ by lookup in $\Pol$. $\mathsf{PrvR}(b,d,q)$ and $\mathsf{PrvS}(b,d,q,r)$ record that the provider at $d$ receives $b$ and answers with token $r$, and $\mathsf{Fin}(b,q,r)$ that the client accepts $r$ for its own request $q$ after the authenticity check on the relayed response. The sidecar issues a fresh nonce $n$ per session and a fresh request identifier $q$ per release, and a provider answers with a fresh response token $r$ bound to $q$.

\heading{Bridging invariants}
Five invariants connect the construction to the operational properties of \S\ref{sec:formal-thms}, formalizing the ordering and authority those theorems assume the loaded code enforces. They are stated as follows.

\begin{table}[b]
\centering
\caption{Falsification variants. Each removes one check; the right column lists the properties whose proofs break, and all other proofs still verify.}
\label{tab:mutations}
\footnotesize
\renewcommand{\arraystretch}{1.12}
\begin{tabular}{@{}ll@{}}
\toprule
\multicolumn{1}{@{}c}{\textbf{Removed check}} & \multicolumn{1}{c@{}}{\textbf{Proofs that break}} \\
\midrule
Oracle refuses pinned measurement & P1 (agreement), P2, P3, P4, P5 \\
Channel key bound in evidence & P1 (agreement), P2, P3, P4, P5 \\
Nonce bound in evidence & P1 (agreement), P2 \\
Destination resolved in baked table & P2, P3, P4, P5 \\
Provider response authenticated & P4 \\
Body-free control channel & P5 (equivalence) \\
\bottomrule
\end{tabular}
\end{table} 

\noindent\textbf{P1, attested release.}
\begin{align*}
\mathsf{Vrf}(m^{\ast},k,n) &\Rightarrow_{\mathrm{inj}} \mathsf{Qt}(m^{\ast},k,n),\\
\mathsf{Rel}(b,k,n,q) &\Rightarrow_{\mathrm{inj}} \mathsf{Vrf}(m^{\ast},k,n).
\end{align*}
The first rules out evidence the platform never issued for the pinned image; the second rules out release without verification, and its injectivity rules out replay, since a captured document justifies no second release.

\noindent\textbf{P2, check-to-use binding.}
\[
\mathsf{Acc}(b^{\ast},k,n,q) \Rightarrow_{\mathrm{inj}} \mathsf{Rel}(b^{\ast},k,n,q).
\]
No party interposes between the check and the use, and no released request is accepted twice. The property is stated for bodies a verifying sidecar protects; the enclave also serves the host's own client sessions, whose bodies the host already knows.

\noindent\textbf{P3, destination authority.}
\[
\mathsf{PrvR}(b^{\ast},d,q) \Rightarrow \exists p.\ \mathsf{Res}(p,d) \wedge d \in \Pol.
\]
A name the host mints resolves to nothing, so the host chooses among destinations but never invents one.

\noindent\textbf{P4, response provenance.}
\[
\mathsf{Fin}(b^{\ast},q,r) \Rightarrow \exists d \in \Pol.\ \mathsf{PrvR}(b^{\ast},d,q) \wedge \mathsf{PrvS}(b^{\ast},d,q,r).
\]

\noindent\textbf{P5, secrecy and a body-free host view.} The adversary never derives the body, and its entire view, the two control-channel messages included, is observationally equivalent under a change of the body, so no host-visible field is body-derived.

\begin{theorem}[Symbolic invariants]
\label{thm:protocol}
In the symbolic model above, with ideal cryptography, unforgeable platform signatures, honest providers at every destination listed in $\Pol$, and the platform soundness assumption of \S\ref{sec:threat}, properties P1 through P5 hold against a host adversary that controls every channel, speaks the control protocol, chooses policy names and credentials, and holds an attestation oracle for every measurement other than the pinned one.
\end{theorem}

\noindent\emph{Proof.} Machine-checked with ProVerif. The five properties compile into ten queries, correspondence assertions for P1 through P4 with the agreement and binding forms injective, and a secrecy query plus an observational-equivalence query for P5; all ten verify. Three further queries check the model rather than the protocol: two reachability queries confirm the honest pipeline completes, so the correspondences are not vacuously true, and one documents the single disclosure by design, that the host learns the gateway credential $g$ it issued (\S\ref{sec:threat}).

\heading{Falsification}
The model also detects the attacks it is meant to rule out. We re-ran all queries on six variants, each removing exactly one protocol check, and Table~\ref{tab:mutations} lists what breaks. The two checks that anchor the attested channel, and the baked-table resolution, fail broadly, body secrecy included; each narrower check breaks exactly the properties it carries.

\heading{Abstractions}
The model abstracts the enclave platform, the quote-verification library, the cipher layer, memory safety, provider behavior, and side channels, and it treats byte-for-byte faithfulness and the control-channel field schemas as assumptions, which the computational theorems of \S\ref{sec:formal} carry as named hardware and faithfulness assumptions.

\end{document}